%% file: main.tex
\documentclass[a4paper,UKenglish,cleveref, autoref, thm-restate]{lipics-v2021}

\input{header}

\title{Approximating $k$-Edge-Connected Spanning Subgraphs via a
  Near-Linear Time LP Solver}
\titlerunning{Approximating $k$-ECSS via a Near-Linear Time LP Solver}
\author{Parinya	Chalermsook}{Aalto University, Espoo, Finland}{	parinya.chalermsook@aalto.fi}{}{}
\author{Chien-Chung	Huang}{École Normale Supérieure, Paris, France}{
  Chien-Chung.Huang@ens.fr	}{}{}
\author{Danupon	Nanongkai}{University of Copenhagen, Copenhagen, Denmark \and KTH,  Stockholm, Sweden}{danupon@gmail.com}{}{}
\author{Thatchaphol Saranurak}{University of Michigan, Ann Arbor, MI, USA}{thsa@umich.edu}{}{}%
\author{Pattara Sukprasert}{Northwestern University, Evanston, IL, USA} {pattara.sk127@gmail.com}{}{}
\author{Sorrachai Yingchareonthawornchai}{Aalto University, Espoo, Finland}{sorrachai.yingchareonthawornchai@aalto.fi}{}{} 
\authorrunning{P. Chalermsook et al.} 

\Copyright{Parinya Chalermsook, Chien-Chung Huang, Danupon Nanongkai,
  Thatchaphol Saranurak, Pattara Sukprasert and Sorrachai Yingchareonthawornchai} 



\category{} 

\relatedversion{} 

\funding{ This project has received funding from the European
  Research Council (ERC) under the European Union's Horizon 2020
  research and innovation programme under grant agreement No 759557 \&
  715672. Nanongkai and Saranurak were also partially supported by the
  Swedish Research Council (Reg. No. 2015-04659.). Chalermsook was
  also supported by the Academy Research Fellowship (grant number
  310415). }

\acknowledgements{ We thank the 2021 Hausdorff
  Research Institute for Mathematics Program Discrete
  Optimization during which part of this work was developed. Parinya Chalermsook thanks Chandra Chekuri for
  clarifications of his FOCS 2017 paper and for giving some
  pointers. We thank Corinna Coupette for bringing
  \cite{GoemansGPSTW94} to our attention.} 

\nolinenumbers 

 \ArticleNo{1}

\begin{document}

\maketitle

\input{abstract.tex}

\input{intro.tex}
\section{Overview of Techniques}
\input{overview.tex}
\section{Preliminaries}
\input{prelim.tex}

\section{Range Mapping Theorem} %

\input{rangemapping}

\section{Fast Approximate LP Solver} 
\input{algo} 
\section{LP Rounding for $k$ECSS (Proof of \Cref{thm: fast rounding})}\label{sec: rounding}
\input{rounding.tex}

\bibliography{references}

\newpage
 \appendix
 \section{Truncated Lazy MWU Increment (Proof of \Cref{thm:tlmi})}\label{sec:lazy MWU increment}

\input{datastructures}

 \section{Omitted Proofs in \Cref{sec: rounding}} \label{sec:omitted-rounding}
 \input{omitted-rounding.tex}

 \input{mwu-covering}

\end{document}

%% file: header.tex
\input{commands}

\newcommand{\support}{\operatorname{support}}

\newcommand{\cC}{\mathcal{C}\xspace}

\newcommand{\eps}{\epsilon}
\newcommand{\dset}{{\mathcal D}}

\newcommand{\cset}{{\mathcal C}}

\newcommand{\supp}{\textsf{supp}}
\newcommand{\vecw}{\textbf{w}}
\newcommand{\vecv}{\textbf{v}}
\newcommand{\vecf}{\textbf{f}}

\newcommand{\ones}{\mathbf{1}}
\newcommand{\minrow}{\textsc{MinRow}}

\newcommand{\winit}{w^{\operatorname{init}}}
\newcommand{\vecwinit}{\vecw^{\operatorname{init}}}
\newcommand{\textsol}{\operatorname{sol}}
\newcommand{\textbest}{\operatorname{best}}
\newcommand{\norm}[1]{\left\lVert#1\right\rVert}
\newcommand{\what}{\widehat{\textbf{w}}}

\newcommand{\phireal}{\phi^{\operatorname{mwu}}}

\newcommand{\wreal}{\vecw^{\operatorname{mwu}}}
\newcommand{\vreal}{\vecv^{\operatorname{mwu}}}

\newcommand{\wapx}{\vecw}

\newcommand{\wsol}{\vecw^{\operatorname{sol}}}
\newcommand{\wbest}{\vecw^{\operatorname{best}}}

\newcommand{\ot}{\tilde{O}}
\renewcommand{\epsilon}{\varepsilon}

\newcommand{\polylog}{\operatorname{polylog}}

\usepackage[procnumbered,ruled,vlined,linesnumbered]{algorithm2e}

\SetCommentSty{mycommfont} 

\newcommand{\bset}{{\mathcal B}}
\newcommand{\val}{{\sf val}}

\newcommand{\opt}{{\sf OPT}}
\newcommand{\mincut}{{\sf mincut}}

\declaretheorem[numberlike=theorem]{invariant}

\crefname{algorithm}{Algorithm}{Algorithms}
\Crefname{algorithm}{Algorithm}{Algorithms}

\DontPrintSemicolon
\SetKw{KwAnd}{and}
\SetProcNameSty{textsc}
\SetFuncSty{textsc}

\newcommand{\cg}{{\sf cong}}


%% file: commands.tex
\allowdisplaybreaks[1]

\makeatletter
\g@addto@macro\bfseries{\boldmath}
\makeatother

\makeatletter
\g@addto@macro\mdseries{\unboldmath}
\g@addto@macro\normalfont{\unboldmath}
\g@addto@macro\rmfamily{\unboldmath}
\g@addto@macro\upshape{\unboldmath}
\makeatother

%% file: abstract.tex
\begin{abstract}
	In the $k$-edge-connected spanning subgraph ($k$ECSS) problem, our goal is to compute a minimum-cost sub-network that is resilient against up to $k$ link failures: Given an $n$-node $m$-edge graph with a cost function on the edges, our goal is to compute a minimum-cost $k$-edge-connected spanning subgraph. This NP-hard problem generalizes the  minimum spanning tree problem and is the ``uniform case'' of a much broader class of survival network design problems (\textsf{SNDP}).
	A factor of two has remained the best approximation ratio for polynomial-time algorithms for the whole class of \textsf{SNDP}, even for a special case of $2$ECSS.
	The fastest $2$-approximation algorithm is however rather slow, taking $O(mn k)$ time [Khuller, Vishkin, STOC'92]. A faster time complexity of $O(n^2)$ can be obtained, but with a higher approximation guarantee of $(2k-1)$  [Gabow, Goemans, Williamson, IPCO'93].

	Our main contribution is an algorithm that  $(1+\epsilon)$-approximates the optimal fractional solution in $\tilde O(m/\epsilon^2)$ time (independent of $k$), which can be turned into a $(2+\epsilon)$ approximation algorithm that runs in time $\tilde O\left(\frac{m}{\epsilon^2} + \frac{k^2n^{1.5}}{\epsilon^2}\right)$ for (integral) $k$ECSS; this improves  the running time of the aforementioned results while keeping the approximation ratio arbitrarily close to a factor of two.

\end{abstract}

%% file: intro.tex
\section{Introduction}

In the $k$-Edge-Connected Spanning Subgraph problem ($k$ECSS), we are given an undirected $n$-node $m$-edge graph $G=(V,E)$ together with edge costs, and want to find a minimum-cost $k$-edge connected spanning subgraph.\footnote{\label{foot:intro:sub-multigraph}Note that this problem should not be confused with a variant that allows to pick the same edge multiple time, which is sometimes also called $k$ECSS (e.g.,~\cite{ChekuriQ17}). We follow the convention in \cite{CzumajL07} and call the latter variant minimum-cost $k$-edge connected spanning sub-multigraph ($k$ECSSM) problem. (See also the work by Pritchard~\cite{Pritchard10}.)}
For $k=1$, this is simply the minimum spanning tree problem, and thus can be solved in $O(m)$ time~\cite{karger1995randomized}.
For $k\geq 2$, the problem is a classical NP-hard problem whose first approximation algorithm was given almost four decades ago, where Frederickson and Jaja \cite{FredericksonJ81} gave a $3$-approximation algorithm that runs in $O(n^2)$ time for the case of $k=2$.
The approximation ratio was later improved to $2$ by an
$\tilde O(mnk)$-time
algorithm of Khuller and Vishkin~\cite{KhullerV94}.\footnote{$\tilde O$ hides $\polylog(n)$ factor.}
This approximation factor of $2$  has remained the best for more than 30 years, even for a special case of $2$ECSS called the weighted tree augmentation problem.
When the running time is of the main concern, the fastest known algorithm takes  $O(n^2)$  time at the cost of a significantly higher $(2k-1)$-approximation guarantee, due to Gabow, Goemans, and Williamson~\cite{GabowGW98}.

This above state-of-the-arts leave a big gap between
algorithms achieving the best approximation ratio and the best time complexity. This gap exists even for $k=2$.
In this paper, we improve the running time of both aforementioned algorithms of~\cite{KhullerV94,GabowGW98} while keeping the approximation ratio arbitrarily close to two.
Our main contribution is a near-linear time algorithm that $(1+\epsilon)$-approximates the optimal {\em fractional} solution.

\begin{theorem}\label{thm:intro:fractional $k$ECSS}
For any $\epsilon > 0$, there is a randomized $\widetilde{O}(m/\epsilon^2)$-time algorithm that outputs a $(1+\epsilon)$-approximate {\em fractional} solution for $k$ECSS.
\end{theorem}

Following, in the high-level, the arguments of Chekuri and Quanrud \cite{ChekuriQuanrud18} (i.e. solving the minimum-weight $k$ disjoint arborescences in the style of \cite{KhullerV94} on the support of the sparsified fractional solution), the above fractional solution can be turned into a fast $(2+\epsilon)$-approximation algorithm for the integral version of $k$ECSS.

\begin{corollary}\label{thm:intro:integral $k$ECSS}
For any $\epsilon>0$, there exist
\begin{itemize}
    \item a randomized $\widetilde{O}(m/\epsilon^2)$-time algorithm that estimates the {\em value} of the optimal solution for $k$ECSS to within a factor $(2+\epsilon)$, and

    \item a randomized $\tilde O\left(\frac{m}{\epsilon^2} + \frac{k^2n^{1.5}}{\epsilon^2}\right)$-time algorithm
that produces a feasible $k$ECSS solution of cost at most $(2+\epsilon)$ times the optimal value.
\end{itemize}
\end{corollary}

We remark that the term $\tilde{O}(k^2 n^{1.5})$ is in fact ``tight'' up to the state-of-the-art algorithm for finding minimum-weight $k$ disjoint arborescences.\footnote{More formally, if a minimum-weight union of $k$ edge-disjoint arborescences can be found in time $T(k,m,n)$, then our algorithm would run in time $T(k, kn, n)$. The term $O(k^2 n^{1.5})$ came from Gabow's algorithm~\cite{gabow1995matroid} that runs in time $O(km \sqrt{n} \log (n c_{\max}))$.}

Prior to our results, a sub-quadratic time algorithm was not known even for special cases of $k$ECSS, called
{\em $k$-Edge-Connected Augmentation} ($k$ECA). In this problem, we are given a $(k-1)$-edge-connected subgraph $H$ of a graph $G$, and we want to minimize the total cost of adding edges in $G$ to $H$ so that $H$ becomes $k$-edge connected. It is not hard to see that if we can $\alpha$-approximates $k$ECSS, then we can $\alpha$-approximates $k$ECA by assigning cost 0 to all edges in $H$. This problem previously admits a $O(kn^2)$-time $2$-approximation algorithm for any even integer $k$~\cite{khuller1993approximation}\footnote{In Khuller and Vishkin~\cite{khuller1993approximation}, the $k$ECA problem aims at augmenting the connectivity from $k$ to $(k+1)$ (but for us it is from $(k-1)$ to $k$.)}. The approximation ratio of 2 remains the best even for 2ECA. Our result in \Cref{thm:intro:integral $k$ECSS} improves the previously best time complexity by a $\tilde{\Theta}(\sqrt{n})$ factor.

\textbf{Perspective.} The gap between algorithms with best approximation ratio and best time complexity in fact reflects a general lack of understanding on {\em fast approximation algorithms}. While polynomial-time algorithms were perceived by many as efficient, it is not a reality in the current era of large data, where it is nearly impossible to take $O(n^3)$ time to process a graph with millions or billions of nodes. Research along this line includes algorithms for sparsest cut \cite{KhandekarRV09,KhandekarKOV2007cut,Sherman09,Madry10-jtree}, multi-commodity flow \cite{GargK07,Fleischer00,Madry10}, and travelling salesman problem \cite{ChekuriQ17,ChekuriQuanrud18}. Some of these algorithms have led to exciting applications such as fast algorithms for max-flow \cite{Sherman13}, dynamic connectivity \cite{NanongkaiSW17,ChuzhoyGLNPS19,SaranurakW19,Wulff-Nilsen17,NanongkaiS17dynamic}, vertex connectivity \cite{LiNPSY21} and maximum matching \cite{BrandLNPSSSW20}.

The $k$ECSS problem belongs to the class of survivable network design problems (SNDPs), where the goal is to find a subgraph ensuring that every pair of nodes $(u,v)$ are $\kappa(u,v)$-edge-connected for a given function $\kappa$. ($k$ECSS is the {\em uniform} version of SNDP where $\kappa(u,v)=k$ for every pair $(u,v)$.)
These problems typically focus on building a network that is resilient against device failures (e.g. links or nodes), and are arguably among the most fundamental problems in combinatorial optimization.
Research in this area has generated a large number of beautiful algorithmic techniques during the 1990s, culminating in the result of Jain~\cite{Jain01} which gives a $2$-approximation algorithm for the whole class of SNDPs. Thus, achieving a {\em fast} $2$-approximation algorithm for SNDPs is a very natural goal.

Towards this goal and towards developing fast approximation algorithms in general, there are two common difficulties:
\begin{enumerate}
    \item Many approximation algorithms inherently rely on solving a {\em linear program (LP)} to find a fractional solution, before performing  rounding steps. However, the state-of-the-art general-purpose linear program solvers are still quite slow, especially for $k$ECSS and SNDP where the corresponding LPs are {\em implicit}.%

    In the context of SNDP, the state-of-the-art (approximate) LP solvers still require at least quadratic time: Fleischer~\cite{fleischer2004fast} designs an $\tilde{O}(m n k)$ for solving $k$ECSS LP, and more generally for SNDP and its generalization~\cite{fleischer2004fast,FeldmannKPS16} with at least $\Theta(m \min\{n,k_{\max}\})$ iterations  of minimum cost flow's computation  are the best known running time where $k_{\max}$ is the maximum connectivity requirements.

    \item Most existing techniques that round fractional solutions to integral ones are not ``friendly'' for the design of fast algorithms. For instance, Jain's celebrated iterative rounding~\cite{Jain01} requires solving the LP  $\Omega(m)$ times.  Moreover, most LP-based network design algorithms are fine-tuned to optimize approximation factors, while designing near-linear time LP rounding algorithms requires limiting ourselves to a relatively small set of tools, about which we currently have very limited understanding.

\end{enumerate}

This paper completely resolves the first challenge for $k$ECSS  and manages to identify a fundamental bottleneck of the second challenge.

\textbf{Challenges for LP Solvers.} Our main challenge is handling the so-called {\em box constraints} in the LPs. To be concrete, below is the LP relaxation of $k$ECSS on graph $G=(V,E)$.
\begin{align}
\min \{\sum_{e \in E} c_e x_e: \sum_{e \in \delta_G(S)} x_e \geq k\ (\forall S \subseteq V), x \in [0,1]^E \} \label{eq:intro:$k$ECSS LP}
\end{align}
where $\delta_G(S)$ is the set of edges between nodes in $S$ and $V\setminus S$.
The box constraints refer to the constraints  $x \in [0,1]^E$.
Without these constraints, we can select the same edge multiple times in the solution; this problem is called $k$ECSSM in \cite{CzumajL07} (see \Cref{foot:intro:sub-multigraph}). Removing the box constraints often make the problem significantly easier. For example, the min-cost $st$-flow problem without the box constraints become computing the shortest $st$-path, which admits a much faster algorithm.

For $k$ECSS, it can be shown that solving \eqref{eq:intro:$k$ECSS LP} without the box constraints can be reduced to solving \eqref{eq:intro:$k$ECSS LP} {\em with $k=1$} and multiplying all $x_e$ with $k$. In other words, without the box constraints, fractional $k$ECSS is equivalent to fractional 1ECSS.%
 This \emph{fractional} 1ECSS can be $(1+\epsilon)$-approximated in near-linear time by plugging in the dynamic minimum cut data structure of Chekuri and Quanrud ~\cite{ChekuriQ17} to the multiplicative weight update framework (MWU).

However, with the presence of box constraints, to use the MWU framework we would need a dynamic data structure for a much more complicated cut problem, that we call, the \textit{minimum normalized free cut} problem (roughly, this is a certain normalization of the minimum cut problem where the costs of up to $k$ heaviest edges in the cut are ignored.)
For our problem, the best algorithm in the static setting we are aware of (prior to this work) is to use Zenklusen's $\tilde O(mn^4)$-time algorithm \cite{Zenklusen14} for the {\em connectivity interdiction} problem.\footnote{In the connectivity interdiction problem, we are given $G=(V,E)$ and $k \in {\mathbb N}$, our goal is to compute $F \subseteq E$ to delete from $G$ in order to minimize the minimum cut in the resulting graph.}
This results in an $\tilde O(kmn^4)$-time static algorithm.  Speeding up and dynamizing this algorithm seems very challenging.
Our main technical contribution is an efficient dynamic data structure (in the MWU framework) for the $(1+\epsilon)$-approximate minimum normalized free cut problem.
We explain the high-level overview of our techniques in Section~\ref{sec:overview}.

\textbf{Further Related Works.} %
The $k$ECSS and its special cases have been studied extensively. For all $k \geq 2$, the $k$ECSS problem is known to be APX-hard~\cite{fernandes1998better} even on bounded-degree graphs~\cite{csaba2002approximability} and when the edge costs are $0$ or $1$~\cite{Pritchard10}.  Although a factor $2$ approximation for $k$ECSS has not been improved for almost $3$ decades, various special cases of $k$ECSS admit better approximation ratios (see for instance~\cite{grandoni2018improved,fiorini2018approximating,adjiashvili2018beating}).
For instance, the unit-cost $k$ECSS ($c_e=1$ for all $e \in E$) behaves very differently, admitting a $(1+ O(1/k))$ approximation algorithm~\cite{gabow2009approximating,laekhanukit2012rounding}.
For the $2$ECA problem, one can get a better than $2$ approximation when the edge costs are bounded~\cite{adjiashvili2018beating,fiorini2018approximating}. Otherwise, for general edge costs, the factor of $2$ has remained the best known approximation ratio even for the $2$ECA problem.

The $k$ECSS problem in special graph classes have also received a lot of attention.
In Euclidean setting, a series of papers by Czumaj and Lingas led to a near-linear time approximation schemes for constant $k$~\cite{czumaj2000fast,czumaj1999approximability}.
The problem is solvable in near-linear time when $k$ and treewidth are constant~\cite{berger2007minimum,chalermsook2018survivable}. In planar graphs, 2ECSS, 2ECSSM and 3ECSSM admit a PTAS~\cite{czumaj2004approximation,borradaile2014polynomial}.

\textbf{Organization.} We provide a high-level overview of our proofs in Section~\ref{sec:overview}.
In Section~\ref{sec:prelim}, we explain the background on Multiplicative Weight Updates (MWU) for completeness (although this paper is written in a way that one can treat MWU as a black box). In~\Cref{sec: range map}, we prove our main technical component. In~\Cref{sec:fast LP solver}, we present our LP solver. In~\Cref{sec: rounding}, we show how to round the fractional solution obtained from the LP solver.
Due to space limitations, many proofs are deferred to Appendix.

%% file: overview.tex
\label{sec:overview}
In this section, we give a high-level overview of our techniques in connection to the known results.
Our work follows the standard Multiplicative Weight Update (MWU) framework together with the Knapsack Covering (KC) inequalities (see Section~\ref{sec:prelim} for more background). Roughly, in this framework, in order to obtain a near-linear time LP solver for $k$ECSS, it suffices to provide a fast dynamic algorithm for a certain optimization problem (often called the oracle problem in the MWU literature):

\begin{definition}[Minimum Normalized Free Cuts]
We are given a graph $G = (V,E)$, weight function $\vecw: E \rightarrow \mathbb{R}_{\geq 0}$, integer $k$, and our goal is to compute a cut $S \subseteq V$ together with edges $F \subseteq \delta_G(S): |F| \leq k-1$ that minimizes the following objective\footnote{For any function $f$, for any subset $S$ of its domain, we define $f(S) = \sum_{s \in S}f(s)$.}:
$$\min_{S \subsetneq V, F \subseteq \delta_G(S): |F| \leq (k-1)} \frac{\vecw(\delta_G(S) \setminus F) }{k - |F|},$$
where $\delta_G(S)$ denotes the set of edges that has exactly one end point in $S$.
We call the minimizer $(S,F)$ the minimum normalized free cut.
\end{definition}
This is similar to the minimum cut problem, except that we are allowed  to ``remove'' up to $(k-1)$ edges (called \emph{free edges}) from each candidate cut $S \subseteq V$, and the cost would be ``normalized'' by a factor of $(k- |F|)$.\footnote{This is in fact a special case of a similar objective considered by Feldmann, K{\"{o}}nemann, Pashkovich and Sanit{\`{a}}~\cite{FeldmannKPS16}, who considered applying the MWU framework for the generalized SNDP}
Notice that there are (apparently) two sources of complexity for this problem. First, we need to find the cut $S$ and second, given $S$, to compute the optimal set $F \subseteq \delta_G(S)$ of free edges. To our best knowledge, a previously fastest algorithm for this problem takes $\tilde{O}(mn^4)$ time by reducing to the connectivity interdiction problem \cite{Zenklusen14}, while we require near-linear time. This is our first technical challenge.

Our second challenge is as follows.
To actually speed up the whole MWU framework, in addition to solving the oracle problem statically efficiently, we further need to implement a dynamic version of the oracle with $\polylog(n)$ update time. In our case,
the goal is to maintain a dynamic data structure on  graph $G=(V,E)$, weight function $\vecw$, cost function $c$, that supports the following operation:

\begin{definition}
The {\sc PunishMin} operation  computes a $(1+O(\epsilon))$-approximate normalized free cut and multiply the weight of each edge $e \in \delta_G(S) \setminus F$ by a factor of at most $e^{\epsilon}$.\footnote{ The actual weight $w(e)$ is updated for all $e \in\delta_G(S) \setminus F$: $w(e) \gets w(e)\cdot \exp(\frac{\epsilon c_{\min}}{c_e})$ where $c_{\min}$ is the minimum edge capacity in $\delta_G(S) \setminus F$.} %
\end{definition}

We remark that invoking the {\sc PunishMin} operation does not return the cut $(S,F)$, and the only change is the weight function $\vecw$ being maintained by the data structure.

\begin{proposition}[Informal]
Assume that we are given a dynamic algorithm that supports {\sc PunishMin} with amortized $\polylog(n)$ cost per operations, then the $k$ECSS LP can be solved in time $\tilde{O}(m)$.
\end{proposition}

Let us call such a dynamic algorithm a fast \textbf{dynamic punisher}. The fact that a fast dynamic punisher implies a fast LP solver is an almost direct consequence of  MWU~\cite{GargK07}.

Therefore, we focus on   designing a fast dynamic algorithm for solving (and punishing) the minimum normalized free cut problem.
Our key idea is an efficient and dynamic implementation of the \textbf{weight truncation} idea.%

\vspace{0.1in}

\fbox{
\begin{minipage}{0.9\textwidth}
\textbf{Weight truncation}: Let $G=(V,E)$ and $\rho \in {\mathbb R}_{\geq 0}$ be a threshold. For any weight function $\vecw$ of $G$, denote by $\vecw_{\rho}$ the truncated weight defined by $\vecw_{\rho}(e) = \min\{\vecw(e), \rho\}$ for each $e \in E$. Call an edge $e$ with $\vecw(e) \geq \rho$ a $\rho$-heavy edge.
\end{minipage}
}

\vspace{0.1in}

Our main contribution is to show that, when allowing $(1+\epsilon)$-approximation, we can use the weight truncation to reduce the minimum normalized free cut to minimum cut with $O(\polylog(n))$ extra factors in the running time. Moreover, this reduction can be implemented efficiently in the dynamic setting.
We present the ideas in two steps, addressing our two technical challenges mentioned above respectively. First, we show how to solve the static version of minimum normalized free cut in near-linear time. Second, we sketch the key ideas to implement them efficiently in the dynamic setting, which can be used in the MWU framework.

We remark that weight truncation technique has been used in different context. For instance, Zenklusen \cite{Zenklusen14} used it for reducing the connectivity interdiction problem to $O(|E|)$ instances of the minimum budgeted cut problem.

\subsection{Step 1: Static Algorithm}
We show that the minimum normalized free cut problem can be solved efficiently in the static setting. For  convenience, we often use the term cut to refer to a set of edges instead of a set of vertices.

Define the objective function of our problem as, for any cut $C$,
$$\val_{\vecw}(C) = \min_{F\subseteq C: |F| \leq k-1} \frac{\vecw(C \setminus F)}{k-|F|}.$$
For any weight function $\vecw$, denote by $\opt_{\vecw} = \min_C \val_{\vecw}(C)$. In this paper, the graph $G$ is always fixed, while $\vecw$ is updated dynamically by the algorithm (so we omit the dependence on $G$ from the notation $\val$ and $\opt$).
When $\vecw$ is clear from context, we sometimes omit the subscript $\vecw$.

We show that the truncation technique can be used to establish a connection between our problem and minimum cut.

\begin{lemma}
\label{lem: baby mapping}
We are given a graph $G=(V,E)$, weight function $\vecw$,  integer $k$, and $\epsilon >0$.
For any threshold $\rho \in (\opt_{\vecw}, (1+\epsilon)\opt_{\vecw}]$,
\begin{itemize}
    \item any optimal normalized free cut in $(G,\vecw)$ is a $(1+\epsilon)$-approximate minimum cut in $(G,\vecw_{\rho})$, and

    \item any minimum cut $C^*$ in $(G,\vecw_{\rho})$ is a $(1+\epsilon)$-approximation for the minimum normalized free cut.
\end{itemize}
\end{lemma}

\begin{proof}
First, consider any cut $C$ with $\val(C) = \opt$.
Let $F \subseteq C$ be an optimal set of  free edges for $C$, so we have $\vecw_{\rho}(C \setminus F) \leq \vecw(C \setminus F) = (k-|F|) \opt$. Moreover, $\vecw_{\rho}(F) \leq |F| \rho$. This implies that
\begin{equation}
\vecw_{\rho}(C) = \vecw_{\rho}(C \setminus F) + \vecw_{\rho}(F) < k \rho
\label{eq:exact map upper}
\end{equation}
Next, we prove that any cut in $(G,\vecw_{\rho})$ is of value at least $k\opt$ (so the cut $C$ is a $(1+\epsilon)$ approximate minimum cut).
Assume for contradiction that there is a cut $C'$ such that $\vecw_{\rho}(C') < k\opt$. Let $F' \subseteq C'$ be the set of $\rho$-heavy edges.
Observe that $|F'| \leq k-1$ since otherwise the total weight $\vecw_{\rho}(C')$ would have already exceeded $k\opt$.
This implies that $\vecw(C' \setminus F') = \vecw_{\rho}(C' \setminus F') < (k-|F'|) \opt$ and that
	$$\val(C') \leq \frac{\vecw(C'\setminus F') }{(k-|F'|)} < \opt$$
	which is a contradiction. Altogether, we have proved the first part of the lemma.

To prove the second part of the lemma, consider a minimum cut $C^*$ in $(G,\vecw_{\rho})$, we have that $\vecw_{\rho}(C^*) < \vecw_{\rho}(C) < k \rho$ (from~\Cref{eq:exact map upper}). Again, the set of heavy edges $F^* \subseteq C^*$ can contain at most $k-1$ edges, so we must have $\vecw(C^* \setminus F^*) < (k-|F^*|) \rho \leq (k-|F^*|)(1+\epsilon) \opt$, implying that $\val(C^*) < (1+\epsilon) \opt$.
%
%
%
%
%
\end{proof}

We remark that this reduction from the minimum normalized free cut problem to the minimum cut problem does not give an exact correspondence, in the sense that a minimum cut in $(G,\vecw_{\rho})$ cannot be turned into a minimum normalized free cut in $(G,\vecw)$. In other words, the approximation factor of $(1+\epsilon)$
is unavoidable.

\begin{theorem}
\label{thm: warmup}
Given a graph $G=(V,E)$ with weight function $\vecw$ and integer $k$, the minimum normalized free cut problem can be $(1+\epsilon)$ approximated by using $O( \frac{1}{\epsilon} \cdot \log n)$ calls to the exact minimum cut algorithm.
\end{theorem}

\begin{proof}
We assume that the minimum normalized free cut of $G$ is upper bounded by some value $M$ which is polynomial in $n = |V(G)|$ (we show how to remove this assumption in \Cref{sec:polynomially bounded cost}). %
For each $i$ such that $(1+\epsilon)^i \leq M$, we compute the minimum cut $C_i$ in $(G,\vecw_{\rho_i})$ 
where $\rho_i = (1+\epsilon)^i$ and return one with minimum value $\val(C_i)$. Notice that there must be some $i^*$ such that $\rho_{i^*} \in (\opt_{\vecw}, (1+\epsilon)\opt_{\vecw}]$ and by the lemma, we must have that $C_{i^*}$ is a $(1+\epsilon)$-approximate solution for the normalized free cut problem.
\end{proof}
By using any  near-linear time minimum cut algorithm e.g.,~\cite{Karger00mincut}, the collorary follows.

\begin{corollary} \label{cor:normalized mincut}
There exists a $(1+\epsilon)$ approximation algorithm for the minimum normalized free cut problem that runs in time $\tilde{O}(|E|/\epsilon)$.
\end{corollary}

\subsection{Step 2: Dynamic Algorithm}

The next idea we use is from Chekuri and Quanrud~\cite{ChekuriQ17}. One of the key concepts there is that it is sufficient to solve a ``range punishing'' problem in near-linear time; for completeness we  prove this sufficiency in Appendix.
In particular, the following proposition is a consequence of their work:

\begin{definition}
A \textbf{range punisher}\footnote{Our range punisher corresponds to an algorithm of
Chekuri and Quanrud~\cite{ChekuriQuanrud18} in one epoch.} is an algorithm that, on any input graph  $G$, initial weight function $\vecw = \vecwinit$, real numbers $\epsilon$, and $\lambda \leq \opt_{\vecwinit}$, iteratively applies {\sc PunishMin} on $(G,\vecw)$  until the optimal becomes at least $\opt_{\vecw} \geq (1+\epsilon)\lambda$.
\end{definition}

The following proposition connects a fast range punisher to a fast LP solver.

\begin{proposition}
\label{prop: CQ epoch}
If there exists a range punisher running in time $$\widetilde{O}\left(|E| + K + \sum_{e \in E}\log ( \frac{\vecw(e)}{\vecwinit(e)})\right)$$ where $K$ is the number of cuts punished, then, there exists a fast dynamic punisher, and consequently the $k$ECSS LP can be solved in near-linear time.
\end{proposition}

This proposition applies generally in the MWU framework independent of problems. That is, for our purpose of solving $k$ECSS LP, we need a fast range punisher for the minimum normalized free cut problem. For Chekuri and Quanrud~\cite{ChekuriQ17}, they need such algorithm for the minimum cut problem (therefore a fast LP solver for the Held-Karp bound).

\begin{theorem}[\cite{ChekuriQ17}, informal]
\label{thm: CQ range punisher}
There exists a fast range punisher for the minimum cut problem.
\end{theorem}

Our key technical tool in this paper is a more robust reduction from the range punishing of normalized free cuts to the one for minimum cuts. This reduction works for all edge weights and is suitable for the dynamic setting.
That is, it is a strengthened version of Lemma~\ref{lem: baby mapping} and is summarized below (see its proof in \Cref{sec: range map}).
\begin{theorem}[Range Mapping Theorem]
\label{lem: full mapping}
Let $(G=(V,E),\vecw)$ be a weighted graph.
Let $\lambda > 0$ and $\rho = (1+\gamma)\lambda$.

\begin{enumerate}

    \item If the value of optimal normalized free cut is in $[\lambda, (1+\gamma)\lambda)$, then the value of minimum cut in $(G,\vecw_{\rho})$ lies in $[k\rho/(1+\gamma), k \rho)$.

    \item For any cut $C$ where $\vecw_{\rho}(C) <k \rho$, then
    $\frac{\vecw(C \setminus F)}{k-|F|} < (1+\gamma) \lambda$ where $F$ contains all $\rho$-heavy edges in $C$. In particular, $\val(C) < (1+\gamma) \lambda$.
\end{enumerate}
\end{theorem}

Given the above reduction, we can implement range punisher fast. We present its full proof in \Cref{sec:fast LP solver} and sketch the argument below.

\begin{theorem}
\label{thm: free cut range punisher}
There exists a fast range punisher for the minimum normalized free cut problem.
\end{theorem}
\begin{proof}(sketch)
We are given $\lambda$ and weighted graph $(G,\vecw): \vecw=\vecwinit$ such that $\opt_{\vecwinit} \geq \lambda$.  Our goal is to punish the normalized free cuts until the optimal value in $(G,\vecw)$ becomes at least $(1+\epsilon)\lambda$. We first invoke \Cref{thm: warmup} to get a $(1+\epsilon)$-approximate cut, and if the solution is already greater than $(1+\epsilon)^2\lambda$, we are immediately done (this means $\opt > (1+\epsilon)\lambda$).

Now, we know that $\opt \leq (1+\epsilon)^2 \lambda \leq (1+3\epsilon) \lambda$.
We invoke Lemma~\ref{lem: full mapping}(1) with $\gamma = 3\epsilon$.
The minimum cut in $(G,\vecw_{\rho})$ has size in the range $[k\rho/ (1+3\epsilon), k \rho)$. We invoke (one iteration of) Theorem~\ref{thm: CQ range punisher} with $\lambda' = k \rho (1+3/\epsilon)$ to obtain a cut $C$ whose size is less than $k \rho$ and therefore, by Lemma~\ref{lem: full mapping}(1), $\val(C) < (1+3\epsilon) \lambda$. This is a cut that our algorithm can punish (we ignore the detail of how we actually punish it -- we would need to do that implicitly since the cut itself may contain up to $m$ edges). We repeat this process until all cuts whose values are relevant have been punished, that is, we continue this process until the returned cut $C$ has size at least $k \rho$.

The running time of this algorithm is
$$\widetilde{O}\left(|E| + K + \sum_{e \in E}\log ( \frac{\vecw_{\rho}(e)}{\vecwinit_{\rho}(e)})\right) \leq \widetilde{O}\left(|E| + K + \sum_{e \in E}\log ( \frac{\vecw (e)}{\vecwinit (e)})\right)$$
Notice that we rely crucially on the property of our reduction using truncated weights.
\end{proof}

We remark that in the actual proof of Theorem~\ref{thm: free cut range punisher}, there are quite a few technical complications (e.g., how to find optimal free edges for a returned cut $C$?), and we cannot invoke Theorem~\ref{thm: CQ range punisher} in a blackbox manner. We refer to \Cref{sec:fast LP solver} for the details. %

\subsection{LP Rounding for $k$ECSS}\label{subsec: LP rounding}

Most known techniques for $k$ECSS (e.g.~\cite{GabowGW98,laekhanukit2012rounding}) rely on iterative LP rounding, which is computationally expensive. 
We achieve fast running time by making use of the 2-approximation algorithm of Khuller and Vishkin~\cite{khuller1993approximation}.

Roughly speaking, this algorithm creates a directed graph $H$ from the original graph $G$ and then compute on $H$ the minimum-weight $k$ disjoint arboresences. The latter can be found by Gabow's algorithms, in either
$\tilde{O}(|E||V| k)$ or $\tilde{O}(k|E| \sqrt{|V|} \log c_{\max})$
time.

To use their algorithm, we will construct $H$ based on the support
of the fractional solution $x$ computed by the LP solver. By the integrality of the arborescence polytope~\cite{schrijver2003combinatorial}, an integral
solution is as good as the fractional solution. However, the support of $x$ can be potentially large, which causes Gabow's algorithm to take longer time. Here our idea is a sparsification of the support, by extending the celebrated sparsification theorem of Benzcur and Karger~\cite{BenczurK15} to handle our problem, i.e., we prove the following (see~\Cref{sec: rounding} for the proofs):

\begin{theorem}
\label{thm: sparsification}
Let $G$ be a graph and $c_G$ its capacities. There exists a capacitated graph $(H, c_H)$ on the same set of vertices that can be computed in
$\ot(m)$
such that (i) $|E(H)| = \ot(n k )$, and (ii) for every cut $S$ and $F \subseteq S: |F| \leq (k-1)$, we have $c_G(S~\setminus~F)~=~(1~\pm~\epsilon)~c_H(S~\setminus~F)$.
\end{theorem}

Benzcur and Karger's theorem corresponds to this theorem when $k = 1$.  We believe that this theorem might have further applications, e.g., for providing a fast algorithm for the connectivity interdiction problem.
Our result implies the following (see Section~\ref{sec: rounding} for the proof):

\begin{theorem} \label{thm: fast rounding}
Assume that there exists an algorithm that finds a minimum-weight $k$-arborescences in an $m$-edge $n$-node graph in time  $T_k(m,n)$. Then there exists a $(2+\epsilon)$ approximation algorithm for $k$ECSS running in time $\tilde{O}(m/\epsilon^2 + T_k(kn/\epsilon^2,n))$
\end{theorem}

Applying \Cref{thm: fast rounding} with the Gabow's algorithm (see~\Cref{thm:fast-karbor} in \Cref{sec: rounding}), we obtain  \Cref{thm:intro:integral $k$ECSS}.

%% file: prelim.tex
\label{sec:prelim}

In this section, we review the multiplicative-weight update (MWU) framework for solving a (covering) LP relaxation of the form $\min \{c \cdot x: A x \geq 1, x \geq 0\}$, where $A$ is an $m$-by-$n$ matrix with non-negative entries and $c \in {\mathbb R}_{\geq 0}^n$. 
Our presentation abstracts away the detail of MWU, so readers should feel free to skip this section. 

Let $A_1, \ldots, A_m$ be the rows of matrix $A$.
Here is a concrete example:
\begin{itemize}
\item \textbf{Held-Karp Bound:} The Held-Karp bound  on input $(G,c)$ aims at solving the LP:\footnote{We refer the readers to~\cite{ChekuriQ17} for more discussion about this LP and Held-Karp bound.}
$$\min \{\sum_{e \in E(G)} c_e x_e: \sum_{e \in S} x_e \geq 2 \mbox{ for any cut $S \subseteq E$}\} $$ %
Matrix $A = A_G$  is a cut-edge incidence matrix of graph $G$ where each row $A_j$ corresponds to a cut $F_j \subseteq E(G)$, so there are exponentially many rows.
Each column corresponds to an edge $e \in E(G)$. There are exactly $|E(G)|$ columns.
The matrix is implicitly given as an input graph $G$.

\end{itemize}

We explain the MWU framework in terms of matrices. Some readers may find it more illustrative to work with concrete problems in mind.

\paragraph*{MWU Framework for Covering LPs:} In the MWU framework for solving covering linear programs, we are given as input an $m$-by-$n$ matrix $A$ and cost vectors $c$ associated with the columns.\footnote{There are several ways to explain such a framework. Chekuri and Quanrud~\cite{ChekuriQ17} follow the continuous setting of Young~\cite{Young14}. We instead follow the combinatorial interpretation of Garg and K\"{o}nemann~\cite{GargK07}.}
Let $\epsilon >0$ be a parameter; that is, we aim at computing a solution $x$ that is $(1+\epsilon)$ approximation of the optimal LP solution.
Denote by $\textsc{MinRow}(A,w)$ the value $\min_{j \in [m]} A_j w$.
We start with an initial weight vector $\textbf{w}^{(0)}_i = 1/c_i$ for $i \in [n]$.
On each day $t=1,\ldots, T$, we compute an approximately ``cheapest'' row $j^*$ such that $A_{j^*} \vecw^{(t-1)} \leq (1+\epsilon) \textsc{MinRow}(A, \vecw^{(t-1)})$, and update the weight ${\bf w}^{(t)}_i \leftarrow \textbf{w}^{(t-1)}_i \exp\left({\frac{\epsilon A_{j^*,i} c_{\min}}{c_i}}\right)$ where $c_{\min} = \min_{i \in [n]} \frac{c_i}{A_{j^*,i}}$.\footnote{In the MWU literature, this is often referred to as an oracle problem.}
After $T = O(n\log n/\epsilon^2)$ many days, the solution can be found by taking the best scaled vectors; in particular, observe that, for any day $t$, the scaled vector $\bar{\textbf{w}}^{(t)}= \textbf{w}^{(t)}/\left(\min_{j \in [m]} A_j \textbf{w}^{(t)}\right)$ is always feasible for the LP.
The algorithm returns $\bar{\textbf{w}}^{(t)}$ which has minimum cost.
The following theorem shows that at least one such solution is near-optimal.

\begin{theorem}
For $T= O(\frac{n\log n}{\epsilon^2})$, one of the solutions $\bar{\vecw}^{(t)}$ for $t \in [T]$ is a $(1+O(\epsilon))$~approximation of the optimal solution $\min \{ c\cdot x: A x \geq 1, x \geq 0\}$.
\label{thm:mwu}
\end{theorem}

Since we use slightly different language than the existing proofs in the literature, we provide a proof in the appendix.

\paragraph*{KC Inequalities:} Our LP is hard to work with mainly because of the mixed packing/covering constraints $x \in [0,1]^n$.
There is a relatively standard way to get rid of the mixed packing/covering constraints by adding Knapsack covering (KC) inequalities into the LP.
In particular, for each row (or constraint) $j \in [m]$, we introduce  constraints:

\[\left(\forall F \subseteq \supp(A_j),|F|\leq (k-1) \right): \sum_{i \in [n] \setminus F} A_{j,i} x_i \geq k- |F| \mbox{, or } \sum_{i \in [n] \setminus F} \frac{A_{j,i}}{(k-|F|)} x_i \geq 1 \]

Let $A^{\textsf{kc}}$ be the new matrix after adding KC inequalities, that is, imagine the row indices of $A^{\textsf{kc}}$ as $(j,F)$ where $j \in [m]$ and $F \subseteq \supp(A_j)$; we define $A^{\textsf{kc}}_{(j,F),i} = A_{j,i}/(k-|F|)$.
The actual number of rows in $A^{\textsf{kc}}$ can be as high as $m \cdot n^{O(k)}$, but our algorithm will not be working with this matrix explicitly.

The following lemma shows that we can now remove the packing constraints. We defer the proof to Appendix.

\begin{lemma}
\label{lem:KC for box}
Any solution to $\{x \in {\mathbb R}^n: A^{kc} x \geq 1, x \geq 0\}$ is feasible for $\{x \in {\mathbb R}^n: A x \geq k, x \in [0,1]\}$. Conversely, for any point $z$ in the latter polytope, there exists a point $z'$ in the former such that $z' \leq z$.
\end{lemma}

\begin{corollary} For any cost vector $c \in {\mathbb R}^n_{\geq 0}$,
$$\min \{c^T x: A^{kc} x \geq 1, x \geq 0\} = \min \{c^T x: A x \geq k, x \in [0,1]\} $$
\end{corollary}

%% file: rangemapping.tex
\label{sec: range map} 

The goal of this section is to prove \Cref{lem: full mapping}, a cornerstone of this paper. We emphasize that it works for \emph{any} weight function $\vecw$.
First, we introduce more notations for convenience. For any cut $C \in \mathcal{C}$, and any subset of edges $F \subseteq E$, we define $\val_{\vecw}(C,F) = \frac{\vecw(C \setminus F)}{k -|F|}$ if $F \subseteq C$ and $|F| < k$; otherwise, $\val_{\vecw}(C,F) = \infty$. Also, denote $\val_{\vecw}(C) = \min_{F \subseteq E}\val_{\vecw}(C,F)$. By definition, we have $\val_{\vecw}(C) = \min_{i \leq k-1} \val_{\vecw}(C,F_i)$ where $F_i$ is the set of heaviest $i$  edges in $C$ with respect to weight function $\vecw$.
We let $\mincut_{\vecw_\rho}$ be the value of a minimum cut with respect with weight $\vecw_\rho$.
When it is clear from context, we sometimes omit the subscript $\vecw$. For any positive number  $\rho$,  let $H_{\vecw,\rho} = \{ e \in E\colon  \vecw(e) \geq \rho\}$ be the set of $\rho$-heavy edges.\footnote{When it is clear from the context, for brevity, we might say that $e$ is a \emph{heavy edge} instead of \emph{$\rho$-heavy edge.}} Define the weight truncation $\vecw_{\rho}(e) = \min\{ \vecw(e), \rho\}$. %

\begin{theorem}[Restatement of \Cref{lem: full mapping}] \label{thm:approx mapping}
We are given a weighted graph $(G,\vecw)$,  $\lambda > 0$ be a parameter and $\rho = (1+\gamma)\lambda$. Then we have the following:
\begin{enumerate}
    \item \label{item:approx mapping1} If $\opt_{\vecw} \in [\lambda, (1+\gamma)\lambda)$, then $\mincut_{\vecw_{\rho}} \in [k\rho/(1+\gamma), k\rho)$, and  %
    \item \label{item:approx mapping2} if a cut $C$ satisfies $\vecw_{\rho}(C) < k\rho$, then $\val_{\vecw}(C, H_{\vecw,\rho} \cap C) < (1+\gamma)\lambda$.
\end{enumerate}
\end{theorem}

Notice that the above theorem not only gives a mapping between solutions of the two problems but also that the heavy edges can be used as a set of free edges. 
We say that a cut $C$ is \textit{interesting} if it contains at most $k-1$ heavy edges, i.e., $|H_{\vecw,\rho} \cap C| < k$.

\begin{proposition} \label{pro:not-interesting-heavy}
If cut $C \subseteq E$ is not interesting (i.e., $|H_{\vecw,\rho} \cap C| \geq k$), then $\val_{\vecw}(C) \geq \rho$ and $\vecw_{\rho}(C) \geq k\rho$.
\end{proposition}
\begin{proof}
The fact that $\vecw_{\rho}(C) \geq k\rho$ follows immediately from the definition of heavy edges.
Let $F_i$ be the set heaviest $i$ edges in $C$ with respect to $\vecw$.  Since $C$ contains at least $k$ heavy edges, we have that for all $i < k$, $C \setminus F_i$ contains at least $k-i$ heavy edges. Therefore, we have $\val_{\vecw}(C)  = \min_{i \leq k-1} \frac{\vecw(C \setminus F_i)}{k-i} \geq \min_{i \leq k-1} \frac{(k-i)\rho}{k-i} = \rho$.
\end{proof} 

\Cref{pro:not-interesting-heavy} says that if a cut is not interesting it must be expensive as a normalized free cut (i.e., high $\val_{\vecw}(C)$) and as a graph cut (i.e., high $w_{\rho}(C)$).
We next give a characterization that relates $\val_{\vecw}$ and the sizes of the cuts for interesting cuts.

\begin{lemma} \label{lem:interesting iff}
Let $C$ be an interesting cut.  Then $\val_{\vecw}(C) \leq \val_{\vecw}(C, H_{\vecw,\rho} \cap C) < \rho$ if and only if $\vecw_{\rho}(C) < k\rho.$
\end{lemma}
\begin{proof}

$(\rightarrow)$
By definition of $\vecw_{\rho}$, we have
\begin{align} \label{eq:def of wrho}
  \vecw_{\rho}(C) = \vecw(C \setminus (H_{\vecw,\rho} \cap C)) + \rho |H_{\vecw,\rho} \cap C|.
\end{align}
If $\val_{\vecw}(C, H_{\vecw,\rho} \cap C  ) < \rho$, then $\vecw(C \setminus H_{\vecw,\rho} \cap C ) < \rho (k -  |H_{\vecw,\rho} \cap C|)$. By \Cref{eq:def of wrho}, we have $\vecw_{\rho}(C) < k\rho$.

$(\leftarrow)$ Denote $F = H_{\vecw,\rho} \cap C $. By definition of $\val$, we have
\begin{align*}
   \val_{\vecw}(C) \leq \val_{\vecw}(C,F) &= \frac{\vecw(C \setminus F)}{k - |F|}
                           \overset{(\ref{eq:def of wrho})}{=} \frac{\vecw_{\rho}(C) - \rho |F|}{k - |F|}     < \frac{k\rho - \rho |F|}{k - |F|} = \rho.
\end{align*}
\end{proof}

\begin{proof} [Proof of \Cref{thm:approx mapping}]
For the first part, we begin by proving that $\mincut_{\vecw_{\rho}} < k\rho$.   
Let $C^*$ be a cut such that $\val_{\vecw}(C^*) = \opt_{\vecw}$. By \Cref{pro:not-interesting-heavy}, $C^*$ must be  interesting.
Since $\val_{\vecw}(C^*) = \opt_{\vecw} < (1+\gamma)\lambda = \rho$,  \Cref{lem:interesting iff} implies that we have $\vecw_\rho(C^*) < k\rho$. Therefore, $\mincut_{\vecw_{\rho}} < k\rho$.

Next, we prove that $\mincut_{\vecw_{\rho}} \geq k\rho/ (1+\gamma)$. Let $C$ be a cut, and denote $F = H_{\vecw,\rho} \cap C $. If $C$ is not interesting, then \Cref{pro:not-interesting-heavy} implies that $\vecw_{\rho}(C) \geq k\rho \geq k\rho/(1+\gamma)$. If $C$ is interesting, by definition of $\vecw_{\rho}$, we have
\begin{align*}
  \vecw_{\rho}(C) = \vecw(C \setminus F) + \rho |F|
\geq \opt_{\vecw} (k - |F|) + \frac{\rho}{1+\gamma} |F|
\geq \frac{\rho k}{1+\gamma}.
\end{align*}
The last inequality follows since by assumption $\opt_{\vecw} \geq \rho/(1+\gamma)$.

For the second part of the theorem, as $\vecw_{\rho}(C) < k\rho$, \Cref{pro:not-interesting-heavy} implies that $C$ is interesting. By \Cref{lem:interesting iff}, $\val_{\vecw}(C, H_{\vecw,\rho}\cap C) < \rho = (1+\gamma)\lambda$.
\end{proof}

%% file: algo.tex
\label{sec:fast LP solver}
In this section, we construct the fast range punisher for the normalized free cut problem. Our algorithm cannot afford to maintain the actual MWU weights, so it will instead keep track of {\em lazy weights}.
From now on, we will use $\wreal$ to denote the actual MWU weights and $\vecw$ the weights that our data structure maintains.

\begin{theorem} [Fast Range Punisher] \label{lem:fast range punisher}
  Given graph $G$ initial weight function $\vecwinit$ and two real values $\lambda,  \eps>0$ such that $\lambda \leq \opt_{\vecwinit}$, there is a randomized algorithm that iteratively applies {\sc PunishMin}  until the optimal with respect to the final weight function $\wreal$ becomes at least $\opt_{\wreal} \geq (1+\epsilon)\lambda$, in time $\widetilde{O}(|E| + K+ \frac{1}{\eps} \sum_{e \in E}\log ( \cdot \frac{\wreal(e)}{\vecwinit(e)}))$, where $K$ is the number of cuts punished.
\end{theorem}

The following theorem is almost standard:  the fast range punisher, together with a fast algorithm for approximating $\opt_w$ for any weight $\vecw$, implies a fast approximate LP solver (e.g., see  \cite{ChekuriQ17,fleischer2004fast}). For completeness, we provide the proof in the Appendix.

\begin{theorem} [Fast LP Solver] \label{thm:fast LP solver}
Given a fast range punisher as described in \Cref{lem:fast range punisher}, and a near-linear time algorithm for approximating $\opt_w$ for any weight function $w$, there is an algorithm  that output $(1+O(\eps))$-approximate solution to $k$ECSS LP in $\ot(m/\eps^2)$ time.
\end{theorem}

Notice that the above theorem implies our main result, \Cref{thm:intro:fractional $k$ECSS}.
The rest of this section is devoted to proving~\Cref{lem:fast range punisher}.
Following the high-level idea of~\cite{ChekuriQ17}, our data structure has two main components:
\begin{itemize}
    \item \textbf{Range cut-listing data structure}: This data structure maintains dynamic (truncated) weighted graph $(G,\vecw_{\rho})$ and  is able to find a (short description of) $(1+O(\epsilon))$-approximate cut whenever one exists, that is, it returns a cut of size between $\lambda$ and $(1+O(\epsilon))\lambda$ for some parameter $\lambda$. Since our weight function $\vecw$ changes over time, the data structure also has an interface that allows such changes to be implemented.
The data structure can be taken  and used directly in a blackbox manner, thanks to~\cite{ChekuriQ17}.

    \item \textbf{Lazy weight data structures}: Notice that a fast range punisher can only afford the running time of $\widetilde{O}\left(\sum_e \log \frac{\wreal(e)}{\vecwinit(e)}\right)$ for updating weights, while in the MWU framework, some edges would have to be updated much more often. We follow the idea of~\cite{ChekuriQ17} to maintain approximate (lazy) weights that do not get updated too often but are still sufficiently close to the real weights.
We remark that $\wreal$ only depends on the sequence of cuts, that {\sc PunishMin} actually punishes.
This lazy weight data structure is responsible for maintaining $\vecw$ that satisfies the following invariant:
\begin{invariant}  \label{inv:new approx weight}
We have $(1-\epsilon)\wreal \leq \vecw \leq \wreal$. 
 \end{invariant}
That is, we allow $\vecw$ to underestimate weights, but they cannot deviate more than by a factor of $(1 - \epsilon)$. In this way, our data structure only needs to update the weight implicitly and output necessary increments to the cut listing data structure whenever the invariant is violated.

\end{itemize}

In sum, our range punisher data structures deal with three weight functions $\vecw$ (lazy weights), $\vecw_{\rho}$ (truncated lazy weights, used by the range cut listing data structure) and $\wreal$ (actual MWU weights, maintained implicitly).

The rest of this section is organized as follows.
In~\Cref{sec:compact}--\Cref{sec:lazy weight},  we explain the components that will be used in our data structure, and in~\Cref{sec:implementationPunisher}, we prove Theorem~\ref{lem:fast range punisher} using these components.

\subsection{Compact representation of cuts}
\label{sec:compact}
\newcommand{\fset}{{\mathcal F}}

This part serves as a ``communication language'' for various components in our data structure.
Since a cut can have up to $\Omega(m)$ edges, the data structure cannot afford to describe it explicitly. We will use a compact representation of cuts~\cite{ChekuriQ17}, which allows us to describe any $(1+\epsilon)$-approximate solution in a given weighted graph using $\ot(1)$ bits; notice that, in the MWU framework, we only care about (punishing) near-optimal solutions, so it is sufficient for us that we are able to concisely describe such cuts.

Formally, we say that a family $\fset$ of  subsets of edges is $\epsilon$-canonical for $(G,\vecw)$ if (i) $|\fset| \leq \widetilde{O}(|E|)$, (ii) any $(1+\epsilon)$-approximate minimum cut of $(G,\vecw)$ is a disjoint union of at most $\ot(1)$ sets in $\fset$, (iii) any set $S \in \fset$ can be described concisely by $\ot(1)$ bits, and (iv) every edge in the graph belongs to $\ot(1)$ sets in $\fset$.
It follows that any $(1+\epsilon)$-approximate cut admits a short description.
Denote by $[[S]]$ a short description of cut $S \in \fset$, and for each $(1+\epsilon)$ approximate cut $C$, $[[C]]$ a short description of $C$.

\begin{lemma}[implicit in \cite{ChekuriQ17}] \label{lem:canonical cuts}
There exists a randomized data structure that, on input $(G,\vecw)$, can be initialized in near-linear time,  (w.h.p) constructs an $\epsilon$-canonical family $\fset \subseteq 2^{E(G)}$, and handles the following queries:
\begin{itemize}

    \item Given a description $[[C]]$ of a $(1+\epsilon)$-approximate cut, output a list of $\ot(1)$ subsets in $\fset$ such that $C$ is a disjoint union of those subsets in $\ot(1)$ time.

    \item Given a description of $[[S]]$, $S \in \fset$, output a list of edges in $S$ in $\widetilde{O}(|S|)$ time.

\end{itemize}
\end{lemma}

\subsection{Range Cut-listing Data Structure}
\label{sec:cut listing}

The cut listing data structure is encapsulated in the following theorem.

\begin{theorem} [Range Cut-listing Data Structure \cite{ChekuriQ17}] \label{thm:cut listing}
The cut-listing data structure, denoted by $\mathcal{D}$, maintains dynamically changing weighted graph $(G,\what)$ and supports the following operations.
\begin{itemize}
\item $\mathcal{D}$.\textsc{Init}$(G,\vecwinit,\lambda, \epsilon)$ where $G$ is a graph, $\what$ is an initial weight function, and $\mincut_{\what} \geq \lambda$: initialize the data structure and  the weight $\what \leftarrow \vecwinit$ in $\ot(m)$ time.

     \item $\mathcal{D}$.\textsc{FindCut}$():$ output either a short description of a $(1+O(\epsilon))$-approximate mincut $[[C]]$ or $\emptyset$ (when $\mincut_{\what} > (1+\epsilon) \lambda$). The operation takes amortized $\ot(1)$ time.

    \item $\mathcal{D}$.\textsc{Increment$(\Delta)$} where $\Delta = \{(e,\delta_e)\}$ is the set of increments (defined by a pair of an edge $e \in E$ and a value $\delta_e \in \mathbb{R}_{\geq 0}$): For each $(e,\delta_e) \in \Delta$, $\what(e) \gets \what(e) + \delta_e$.   The operation takes $\ot(|\Delta|)$ time (note that $|\Delta|$ corresponds to the number of increments).
\end{itemize}
\end{theorem}

As outlined earlier, the cut listing data structure will be invoked with $\what = \vecw_{\rho}$.

\subsection{Truncated Lazy MWU Increment}
\label{sec:lazy weight}

The data structure is formally summarized by the definition below.

\begin{definition}[Truncated Lazy MWU Increment]
 A truncated lazy MWU increment denoted by $\mathcal{L}$ maintains the approximate weight function $\vecw$ explicitly, and exact weight $\wreal$ implicitly and  supports the following operations:\footnote{This is implicit in the sense that $w$ is divided into parts and they are internally stored in different memory segments. Whenever needed, the real weight can be constructed from the memory content in near-linear time.}

 \begin{itemize}
     \item   $\mathcal{L}.\textsc{Init}(G,\vecwinit,\rho)$ where $G$ is a graph, $\vecwinit$ is the initial weight function, $\rho \in \mathbb{R}_{>0}$: Intialize the data structure, and set $\vecw \leftarrow \vecwinit$.

     \item $\mathcal{L}.\textsc{Punish}([[C]])$ where $C$ is a cut: Internally punish the free cut $(C,F)$ for some $F$ (to be made precise later) and output a list of increment $\Delta = \{(e,\delta_e)\}$ so that for each $e \in E$, $\vecwinit(e)$ plus the total increment over $e$ is $\vecw_{\rho}(e)$.

     \item  $\mathcal{L}.\textsc{Flush}()$: Return the exact weight $\wreal$.
 \end{itemize}
\end{definition}

Remark that the output list of increments returned by {\sc Punish} is mainly for the purpose of syncing with the cut listing data structure (so it aims at maintaining $\vecw_{\rho}$ instead of $\vecw$).  Also, in the {\sc Punish} operation, the data structure must compute the set $F \subseteq C$ of free edges efficiently (these are the edges whose weights would not be increased).
 This is one of the reasons for which we cannot use the lazy update data structure in~\cite{ChekuriQ17} as a blackbox. Section~\ref{sec:dataStructures} will be devoted to proving the following theorem.

\begin{theorem} \label{thm:tlmi}
There exists a lazy MWU increment with the following time complexity: (i) init operation takes $\ot(m)$ time, (ii)  \textsc{Punish} takes $\ot(K) +\widetilde{O}\left(\sum_e \log \frac{\wreal(e)}{\vecwinit(e)}\right)$ time in total where $K$ is the number of calls to $\textsc{Punish}$ and outputs at most $\widetilde{O}\left(\sum_e \log \frac{\wreal(e)}{\vecwinit(e)}\right)$ increments, and (iii) flush takes $\ot(m)$ time.   Moreover, the Invariant~\ref{inv:new approx weight} is maintained throughout the execution.
\end{theorem}

\subsection{A Fast Range Punisher for Normalized Free Cut Problem}

\label{sec:implementationPunisher}

\newcommand{\calT}{\mathcal{T}}
\newcommand{\calD}{\mathcal{D}}
\newcommand{\calL}{\mathcal{L}}

\newcommand{\lset}{\mathcal{L}}

Now we have all necessary ingredients to prove~\Cref{lem:fast range punisher}. The algorithm is very simple and described in~\Cref{alg:fast range punisher}. We initialize the cut-listing data structure $\dset$ so that it maintains the truncated weight $\vecw_{\rho}$ and the lazy weight data structure $\lset$.  
We iteratively use $\dset$ to find a cheap cut in $(G,\vecw_{\rho})$ until no such cut exists. Due to our mapping theorem, such a cut found can be used for our problem, and the data structure $\lset$ is responsible for punishing the weights (Line 8) and returns the list of edges to be updated (this is for the cut-listing $\dset$ to maintain its weight function $\vecw_{\rho}$).

\subsubsection*{Algorithm}
\begin{algorithm}[H]
\KwIn{$G,\vecwinit,\lambda , \epsilon$ such that $\opt_{\vecwinit}\geq \lambda$.}
\KwOut{a correct weight function $\vecw = \wreal$ such that $\opt_{\vecw} \geq (1+\eps)\lambda$.}
\BlankLine
$\vecw \gets \vecwinit$ and $\rho \gets (1+\eps)\lambda$\;
Let $\vecw_{\rho}$ be the truncated weight function of $\vecw$.\;
\lIf{ $\mincut_{\vecw_{\rho}} \geq k\rho$}{ \Return{$\vecw$.}} \label{line:early return}
  Let $\calD$ and $\calL$ be cut listing data structure, and truncated lazy MWU increment. \;
  $\calD.\textsc{Init}(G,\vecw_{\rho}, k\rho/(1+\eps), \epsilon)$\;
  $\calL.\textsc{Init}(G,\vecw,\rho,\eps)$\;
  \While{$\calD.\textsc{FindCut}()$ \normalfont{returns} $[[C]]$}
  {
    $\Delta \gets \calL.\textsc{Punish}([[C]])$\;
    $\calD.\textsc{Increment}(\Delta)$\;
  }
  $\vecw \gets \calL.\textsc{Flush}()$\;

  \Return{$\vecw$.} \label{line:late return}
\caption{\textsc{FastRangePunisher}($G,\vecw,\lambda$)}
\label{alg:fast range punisher}
\end{algorithm}

\subsubsection*{Analysis}
By input assumption, we have $\opt_w \geq \lambda$. If $\vecw$ is returned at line \ref{line:early return}, then $\mincut_{\vecw_{\rho}} \geq k\rho$. By \Cref{thm:approx mapping}(\ref{item:approx mapping1}), $\opt_{\wreal} \geq \opt_{\vecw} \geq \rho= (1+\eps)\lambda$, and we are done (since minimum cut can be computed in near-linear time). Now, we assume that $\vecw$ is returned at the last line. The following three claims imply \Cref{lem:fast range punisher}.

\begin{claim} For every cut $[[C]]$ returned by the range cut listing data structure during the execution of \Cref{alg:fast range punisher}, we have that $(C,H_{\vecw,\rho} \cap C)$ is a $(1+O(\epsilon))$-approximation to $\opt_{\wreal}$ at the time $[[C]]$ is returned.
\end{claim}

We remark that it is important that our cut punished must be approximately optimal w.r.t. the actual MWU weight.

\begin{proof}
  By definition of $\calL.\textsc{Flush}()$ operation, we always have that the exact weight function and approximate weight function are identical at the beginning of the loop.  By definition of $\calL.\textsc{Punish}([[C]])$, the total increment plus the initial weight at the beginning of the loop for every edge $e$ is $\vecw_{\rho}(e)$ and \Cref{inv:new approx weight} holds. Therefore, by definition of $\calD.\textsc{Increment}(\Delta)$, the range cut-listing data structure maintains the weight function $\vecw_{\rho}$ internally. We now bound the approximation of each cut $[[C]]$ that $\calD.\textsc{FindCut}()$ returned. Let $F =  H_{\vecw,\rho}\cap C$. By definition of $\textsc{FindCut}()$, we have that $\vecw_{\rho}(C) < k\rho$. By \Cref{thm:approx mapping}(\ref{item:approx mapping2}), $\val_{\vecw}(C, F) < (1+\eps)\lambda$. By \Cref{inv:new approx weight}, we have that $\val_{\wreal}(C, \tilde F) < (1+O(\eps))\lambda$.  Since $\opt_{\wreal} \geq \opt_{\winit} \geq \lambda$, we have $(C, F)$ is a $(1+O(\eps))$-approximation to $\opt_{w}$.
\end{proof}

\begin{claim}
At the end of \Cref{alg:fast range punisher},  we have $\opt_{\wreal} \geq (1+\eps)\lambda.$
\end{claim}
\begin{proof}
  Consider the time when $\calD.\textsc{FindCut}()$ outputs $\emptyset$. The fact that this procedure terminates means that $\mincut_{\vecw_{\rho}} \geq k\rho$. Therefore,  \Cref{thm:approx mapping}(\ref{item:approx mapping1})  implies that  $\opt_{\vecw} \geq (1+\eps)\lambda$. Let $(C^*, F^*)$ be an optimal normalized free cut with respect to $\wreal$. We have

  \begin{align*}
    \opt_{\wreal} &= \val_{\wreal}(C^*,F^*) \\
    &\geq \val_{\vecw}(C^*,F^*) \\
    & \geq \opt_{\vecw} \\
    &\geq (1+\eps)\lambda
  \end{align*}

  \noindent where the first inequality follows from \Cref{inv:new approx weight}.
\end{proof}

\begin{claim}
  \Cref{alg:fast range punisher} terminates in $\ot(m+K+\frac{1}{\eps} \cdot \sum_{e \in E}\log ( \frac{\vecw(e)}{\vecwinit(e)}))$ time where $K$ is the number of \textsc{Punish} operations.
\end{claim}

\begin{proof}
We first bound the running time due to truncated lazy MWU increment. By \Cref{thm:tlmi}, the total running time due to $\calL$ (i.e.,  $\calL.\textsc{Init}, \calL.\textsc{Punish}, \calL.\textsc{Flush}$)  is  $\ot(m+K+\frac{1}{\eps} \cdot \sum_{e \in E}\log ( \frac{\vecw(e)}{\vecwinit(e)}))$ time where $K$ is the number of \textsc{Punish} operations.  We bound the running time due to cut-listing data structure. Observe that the number of cuts listed equals the number of calls of \textsc{Punish} operations, and the total number of edge increments in $\calD$ is $\widetilde{O}\left(\frac{1}{\eps} \cdot \sum_{e \in E}\log ( \frac{\vecw(e)}{\vecwinit(e)})\right)$. By \Cref{thm:cut listing}, the total running time due to $\calD$ (i.e, $\calD.\textsc{Init}, \calD.\textsc{FindCut}(), \calD.\textsc{Increment}(\Delta)$) is as desired.
\end{proof}

%% file: rounding.tex
In this section, we show how to round the LP solution $x$ found by invoking \Cref{thm:intro:fractional $k$ECSS}. The main idea is use a sampling 
technique to sparsify the support of $x$. On the subgraph $G'\subseteq G$ based on this sparsified support, we apply the 2-approximation algorithm of Khuller and Vishkin  \cite{KhullerV94} to obtain a ($2+\epsilon$)-approximation solution.

Let $G$ be a graph with capacities $c$ (we omit capacities whenever it is clear from the context). 
Our algorithm performs the following steps. 

\paragraph*{Step 1: Sparsification.} 
We will be dealing with the following LP relaxation for $k$ECSS. 
\[\min \{\sum_{e \in E(G)} c(e) x_e:  \sum_{e \in C \setminus S} x_e \geq k - |S|, \quad \forall C \in \cC \ \forall S \in \{F : \left |F \right | \leq k-1 \wedge F \subseteq C\},  x \geq 0\} \] 
Denote by ${\sf LP}_{kECSS}(G)$ the optimal LP value on input $G$. 
We prove the following lemma in \Cref{sec:spase graph} that will allow us to sparsify our graph without changing the optimal fractional value by too much: 

\begin{lemma}
\label{lem:spase graph} 
Given an instance $(G,c)$,  and in $\ot(m/\eps^2)$ time,  we can compute a subgraph $G'$ having at most $\widetilde{O}(nk/\epsilon^2)$ edges such that ${\sf LP}_{kECSS}(G') = (1\pm O(\epsilon)) {\sf LP}_{kECSS}(G)$. 
\end{lemma} 

The first step is simply to apply this lemma to obtain $G'$ from $G$.

\paragraph*{Step 2: Reduction to $k$-arborescences.} Next, we reduce the $k$ECSS problem to the minimum-cost $k$-arborescence problem which, on capacitated directed graph $(H,c_H)$, can be described as  the following IP: 
\[\min \{\sum_{e \in E(H)} c_H(e) z_e: \sum_{e \in \delta^+(C)} z_e \geq k \mbox{ for $C \in \cset$}; z \in \{0,1\}^{E(H)}\} \] 
where $\cset$ is the set of all cuts $C$ such that $\{r\} \subseteq C \subsetneq V(G)$. 
Denote by $\opt_{ar}(H)$ and ${\sf LP }_{ar}(H)$ the optimal integral and fractional values\footnote{The relaxation is simply $\min \{\sum_{e \in E(H)} c_H(e) z_e: \sum_{e \in \delta^+(C)} z_e \geq k \mbox{ for $C \in \cset$}; z \in [0,1]^{E(H)}\} $} of the minimum-cost $k$-arborescence problem respectively.
We use the following integrality of its polytope: 

\begin{theorem}
[\cite{schrijver2003combinatorial}, Corollary 53.6a]
\label{thm: integrality of ar}  \label{thm:arbor-int}
The minimum-cost $k$-arborescence's polytope is integral, so we have that $\opt_{ar}(H) = {\sf LP}_{ar}(H)$ for every capacitated input graph $H$.  
\end{theorem}

For any undirected graph $G$, denote by $D[G]$ the directed graph obtained by creating, for each (undirected) edge $uv$ in $G$, two  edges $(u \rightarrow v)$ and $(v \rightarrow u)$ in $D[G]$ whose capacities are just $c(uv)$.  
We will use the following theorem by Khuller and Vishkin (slightly modified) that relates the optimal values of the two optimization problems. 

\begin{theorem}
\label{thm: KV} 
For any graph $(H,c)$, the following properties hold: 
\begin{itemize}
    \item ${\sf LP}_{ar}(D[H]) \leq 2 {\sf LP}_{kECSS}(H)$, and 
    \item Any feasible solution for $k$-arborescences in $D[H]$ induces a feasible $k$ECSS solution in $H$ of at most the same cost. 
\end{itemize}
\end{theorem}

Note that \Cref{thm: integrality of ar} and the algorithm by Khuller and Vishkin imply that the integrality gap of the $k$ECSS LP is at most $2$.
While this result is immediate, it seems to be a folklore. To the best of our knowledge, it was not explicitly stated anywhere in the literature.
This integrality gap allows us to obtain the first part of \Cref{thm:intro:integral $k$ECSS}.
We defer the proof of \Cref{thm: KV} to \Cref{sec: proof KV}.
Our final tool to obtain \Cref{thm: fast rounding} (and the second part of \Cref{thm:intro:integral $k$ECSS}) is Gabow's algorithm:

\begin{theorem} [\cite{gabow1995matroid}] \label{thm:fast-karbor}
 Given a graph $G=(V,E,c)$ with positive cost function $c$, a fixed root $r \in V$, and let $c_{\max}$ be the maximum cost on edges, there exists an algorithm that in $\tilde{O}(k m \sqrt{n} \log (n c_{\max}))$ time outputs the integral minimum-cost $k$-arborescence.
\end{theorem}

\paragraph*{Algorithm of \Cref{thm: fast rounding}.} Now, using the graph $G'$ created in the first step, we create $D[G']$, and invoke Gabow's algorithm to compute an optimal $k$-arborescence in $D[G']$. Let $S \subseteq E(G)$ be the induced $k$ECSS solution.

The cost of $S$ is at most:
\begin{align*}
  \opt_{ar}(D[G']) &\leq {\sf LP}_{ar}(D[G']) \\
  &\leq 2 {\sf LP}_{kECSS}(G') \\
  &\leq 2 (1+O(\epsilon)){\sf LP}_{kECSS}(G)\\
  &\leq 2 (1+O(\epsilon)) \opt_{kECSS}(G)
\end{align*}
The first inequality is due to \Cref{thm: integrality of ar}. The second one is due to~\Cref{thm: KV} (first bullet). The third one is due to~\Cref{lem:spase graph}. 

\paragraph*{Analysis} %
Step 1 takes $\ot(m/\epsilon^2)$ time, by \Cref{lem:spase graph}.
As the sparsified $G'$ has $m' = \ot(\frac{nk}{\epsilon^2})$ edges, for Step 2, by \Cref{thm:fast-karbor}, we can compute the arborescence in $O(k m' \sqrt{n} \log (nc_{\max})) = \ot(\frac{k^2 n^{1.5}}{\epsilon^2} \log c_{\max})$ time.
We show in \Cref{subsec:bound_cmax} how to remove the term $\log {c_{\max}}$ in our case.
In summary, the total running time is $\ot\left(  \frac{m}{\epsilon^2}+ \frac{k^2n^{1.5}}{\epsilon^2} \right)$.
Notice that the running time can be $\ot(\frac{m}{\epsilon^2} + T_k(kn/\epsilon^2,n))$ if we let the running time of \Cref{thm:fast-karbor} be $T_k(m,n)$, this complete the proof for \Cref{thm: fast rounding}.

%% file: datastructures.tex
\label{sec:dataStructures}

\subsection{Additive Accuracy}

\newcommand{\calB}{\mathcal{B}}
\newcommand{\calF}{\mathcal{F}}
\newcommand{\cF}{\mathcal{F}}
\newcommand{\calS}{\mathcal{S}}

Notice that, at any time, we always have $\wreal(e)= \frac{1}{c(e)} \cdot
\exp \left( \vreal(e) s(e)\right)$ for some positive real numbers $\vreal(e)$
and $s(e) = \frac{\epsilon}{c(e)}$.
For the true vector $\vreal$, the update rule for $(C, F)$ becomes the
following: $\vreal(e) \leftarrow \vreal(e) + c_{\min}$ for all $e \in C
\setminus F$. This update causes all
edges in $C \setminus F$ to increase their $\vreal(e)$ by the same amount
of $c_{\min}$. By \Cref{thm:approx mapping}(2), it is enough to use
$F$ to be always $H_{\rho, \vecw} \cap C$, i.e., the set of heavy edges with respect to $\vecw$ inside
$C$.  From now on, we always use $H_{\vecw,\rho} \cap C$ as a free
edge set whenever we punish $C$.  

Instead of maintaining the approximate vector $\vecw$ for the real
vector $\wreal$, we instead work
with the additive form of the approximate vector $\vecv$ for the real vector $\vreal$,
and we bound the additive error:
\begin{align} \label{eq:additive error small}
  \forall e \in E, \vreal(e) - \eta/s(e) \leq \vecv(e) \leq \vreal(e) 
\end{align}
Next, we show that it is enough to work on $\vreal$ with additive errors. 
\begin{proposition}
  If \Cref{eq:additive error small} holds, then $\forall e \in E, \wreal(e)
  (1-\eta) \leq \vecw(e) \leq \wreal(e)$. 
\end{proposition}
\begin{proof}
  Fix an arbitrary edge $e \in E$, we have $\vecw(e) \leq
  \wreal(e)$. Moreover,
$$ \vecw(e) \geq \frac{1}{c(e)} \exp( (\vreal(e) - \eta/s(e)) s(e)) =
\frac{1}{c(e)} \cdot \frac{\exp(\vreal(e)s(e))}{\exp(\eta)} \geq (1-\eta) \wreal(e).$$
\end{proof}
\subsection{Local Bookkeeping} \label{sec:local book}

\newcommand{\textdeg}{\operatorname{q}}
\newcommand{\textref}{\operatorname{ref}}
\newcommand{\textlast}{\operatorname{last}}
\newcommand{\textdiff}{\operatorname{diff}}
\newcommand{\textpriority}{\operatorname{priority}}

We describe the set of variables to maintain in order to support
$\textsc{Punish}$ operation efficiently. Let $\calF$ be a
$\eps$-canonical family of subsets of edges (as defined in
\Cref{lem:canonical cuts}). We call each subset of edges in $\calF$ as a canonical cut.  Let $\bar E = E \setminus H_{\vecw, \rho}
$ be the set of non-heavy edges where $H_{\rho, \vecw} = \{e \in E\colon \vecw(e) \geq \rho\}$ is the set of heavy edges. We define a bipartite graph $\calB =
(\calF, \bar E, E_{\calB})$ where the first vertex partition  is the set
of canonical cuts $\calF$ , the
second vertex partition is $\bar E$, and for each $S \in \calF$ and for
each $e \in \bar E$, we add an edge  $(S,e)$ to $E_{\calB}$ if and
only if $e \in S$.  Let $\textdeg(\calB) =$ the maximum degree of
nodes in $\bar E$ in graph $\calB$.  Since $\calF$ is
$\eps$-canonical, $\textdeg(\calB) = \ot(1)$. By \Cref{lem:canonical
  cuts}, given a description
$[[C]]$ of 1 or 2$-$repsecting cut, we can compute a list of at most
$\ot(1)$ canonical cuts in $\calF$ in $\ot(1)$ time.

We maintain the following variables:
\begin{enumerate}
\item For each canonical cut $S \in \calF$,
  \begin{enumerate} 
     \item we have a non-negative
  real number $\textref(S)$ representing the reference point for
  the total increase in $S$ so far.  %
  \item Also, we create a min priority queue
  $Q_{S}$ containing the set of neighbors $N_{\calB}(S)$ (which is the
  set of edges in $\bar E$ that $S$ contains). %
  \item Also, we define $c_{\calB}(S) = \min_{e \in N_{\calB}(S)}c(e)$ for the purpose of
  computing $c_{\min}$  which is the minimum capacity $c(e)$ for all
  edge $e$ in the cut (excluding  heavy edges) that we want to punish. %
  \end{enumerate}
\item For each edge $(S,e) \in E_{\calB}$, we  have a number
  $\textlast(S,e)$ representing the last update point for $e$ in $S$.
 \item For each edge $e \in E$, we maintain $\vecv(e)$.
\end{enumerate}

For each edge $(S,e) \in E_{\calB}$, we define $\textdiff(S,e) =
\textref(S) - \textlast(S,e) \geq 0$.  This difference represents the total
slack from the exact weight of $e$ on $S$ (we will ensure that the
slack is non-negative by being ``lazy''). When summing over all canonical cuts that contains $e$, we
ensure that $\sum_{S \ni  e}\textdiff(S,e) =
\vreal(e) - \vecv(e)$.  More formally, we maintain the following
invariants throughout the execution of the truncated lazy increment.

\begin{invariant}  \label{inv:lazy update}Let $\eta' = \eta/\textdeg(\bset)$. %
  \begin{enumerate}[(a)]
   \item \label{item:total diff not far} for all $e \in \bar E$, $ \frac{\eta}{s(e)}  \geq \sum_{S:
       (S,e) \in E_{\calB}}\textdiff(S,e) = \vreal(e) - \vecv(e)$,  
   \item \label{item:priority set correctly} for all $(S,e) \in E_{\calB}$, $Q_S.\textpriority(e) =
     \textlast(S,e) +  \frac{\eta'}{s(e)}$, and
   \item  \label{item:heavy edges exact}for all $e \in H_{\vecw,\rho}$, $\vreal(e) = \vecv(e)$.
   \end{enumerate}
\end{invariant}

Intuitively, the first invariant means for each $e \in \bar E$, the total difference over all
$S \ni e$ is bounded. The second invariant ensures that $\textref(S)
\leq Q_S.\textpriority(e)$ if and only if 
$\textdiff(S,e)$ is small, and we can apply extract min operations on $Q_S$ to detect all
edges whose priority exceeds the reference point efficiently.  The
third invariant means we restore the exact value for all heavy
edges.

\begin{proposition}
\Cref{inv:lazy update}\ref{item:total diff not far} implies \Cref{eq:additive error small}.
\end{proposition}

Also, this invariant allows us to ``reset'' $\vecv$ to be $\vreal$
efficiently. 

\begin{algorithm}[H]
  \BlankLine
  $\vecv(e) \gets \vecv(e) + \sum_{S: (S,e) \in E_{\calB}} \textdiff(S,e)$ \; %
  \For{\normalfont{each} $S: (S,e) \in E_{\calB}$}{
  $\textlast(S,e) \gets \textref(S)$ \;
   $Q_S.\textpriority(e) \gets \textlast(S,e) + \frac{\eta'}{s(e)}$ \;
 }
\caption{\textsc{Reset}($e$)}
\label{alg:reset}
\end{algorithm}

Since priority queue supports the change of priority in $O(\log m)$ time, we
have:
\begin{proposition}
  The procedure \textsc{Reset} can be implemented in time $O( \textdeg(\calB) \cdot \log m) = \ot(1)$.
\end{proposition}

\subsection{Init}
Define $\vecv = v_0$ where $v_0$ is the additive form of $w_0$. We construct the bipartite graph $\calB = (\calF,
\bar E, E_{\calB})$ as defined in \Cref{sec:local book}. We use \Cref{lem:canonical cuts} to construct
$\calB$ in $\ot(m)$ time. For each $S \in \calF$, we create a min priority
queue $Q_S$ containing all the elements in $N_{\calB}(S)$ where for
each $e \in N_{\calB}(S)$, we set $Q_s.\textpriority(e) =
\eta'/s(e)$. We also define $\textref(S) = 0$ for all $S \in \calF$, and
$\textlast(S,e) = 0$ for all $(S,e) \in E_{\calB}$. By design, the
invariants are satisfied. The total running time of this step is
$\ot(m)$.  
\subsection{Punish}

Given a short description of 1 or 2-respecting cut $[[C]]$, we apply
\Cref{lem:canonical cuts} to obtain a set $\calS \subseteq \calF$ of $\ot(1)$ canonical
cuts whose disjoint union is $C$ in $\ot(1)$ time.  Recall that the
update increases $\vreal(e)$ by $c_{\min}$ for each $e \in C -
H_{\vecw,\rho}$ where $c_{\min} = \min_{e \in C - H_{\vecw,\rho}} c(e)$.
\begin{claim}
We can compute $c_{\min}$ in $\ot(1)$ time.
\end{claim}
\begin{proof}
By definition of $c_{\calB}(S)$,   $\min_{S \in \calS}
c_{\calB}(S) = \min_{S \in \calS} \min_{e \in N_{\calB}(S)}c(e) =  \min_{e \in C -
  H_{\vecw,\rho}} c(e) = c_{\min}$. The claim follows because there are $\ot(1)$ canonical
cuts in $\calS$ and we maintain the value $c_{\calB}(S)$ for every $S \in \calF$.%
\end{proof}

In the first step, for each $S \in \calS$, we set $\textref(S) \gets
\textref(S) + c_{\min}$. This takes $\ot(1)$ time because $|\calS| =
\ot(1)$ and potentially causes a violation to \Cref{inv:lazy
  update}\ref{item:total diff not far}.

In the second step, we check and fix the invariant violation as follows.  For
each $S \in \calS$, let $W_S = \{ e \in S \setminus H_{\vecw,\rho} \colon \textref(S) > Q_S.\textpriority(e)\}$ be the set of all edges
in $S \setminus H_{\vecw,\rho}$ whose priority in $Q_S$ is smaller
than the reference point $\textref(S)$. For each $e \in W_S$, we call
the procedure $\textsc{Reset}(e)$. This take times $O(r \cdot q(\calB)
\log m) = \ot(r)$ where $r$ is the number of calls to $\textsc{Reset}$
procedure. There will be new heavy edges after this step, which means
we need to update $\calB$ to correct the set $\bar E$. 

In the third step, we identify new heavy edges from the set of edges
that we called \textsc{Reset} procedure in the second step, then we
remove each edge in the set from the associated priority queues and
from the graph $\calB$ as follows. Let $U = \bigcup_{S \in \calS}W_S$. Define $U_{H} = \{ e \in U \colon \vecw(e) \geq
\rho\}$.  For each $e \in U_H$, for all $D \in N_{\calB}(e)$, remove $e$ from the priority queue
$Q_D$ and update the value $c_{\calB}(D)$ (to get a new minimum after
removing $e$). Finally, delete all nodes in $U_H$ from $\calB$. The third step
takes $O( |U| + |U_H| \textdeg(\calB) \log m + |U_H| \textdeg(\calB))
= \ot(r)$ time. The running time follows because the $|U| = r$ and $|U_H| \leq |U|$.

Finally, we output $\Delta$ where $\Delta$ is constructed as
follows. For each $e \in U$, let $\vecw'(e)$ be the weight of $e$
before \textsc{Reset}$(e)$ is invoked. If $e \not \in U_H$, then we define $\delta_e =
\vecw(e) - \vecw'(e)$. Otherwise, we define $\delta_e = \rho -
\vecw'(e)$. Then, we add $(e,\delta_e)$ to $\Delta$.

\begin{lemma}
  If \Cref{inv:lazy update} holds before calling $\textsc{Punish}([[C]])$, then
  \Cref{inv:lazy update} holds afterwards.
\end{lemma}
\begin{proof}
  In the first step, we have $ \bigcup_{S \ni \calS} N_{\calB}(S) = C
  \setminus H_{\vecw,\rho}$, and thus the violation to \Cref{inv:lazy update}\ref{item:total
  diff not far} can only happen due to some edge $e \in C
\setminus H_{\vecw,\rho}$.  Because the unions are over disjoint sets, for each
edge $e \in C \setminus H_{\vecw,\rho}$, there is a unique
canonical cut $S_e \in \calS$ such that $N_{\calB}(S_e) \ni e$.

\begin{claim} 
  If there is a violation to \Cref{inv:lazy update}\ref{item:total
  diff not far} due to an edge $e \in C \setminus H_{\vecw,\rho}$, then \textsc{Reset}$(e)$ is invoked in the second step.
\end{claim}
\begin{proof}
    Since  \Cref{inv:lazy update}\ref{item:total
  diff not far} is violated due to an edge $e$, we have $\sum_{S':
  (S',e) \in E_{\calB}} \textdiff(S',e) > \eta/s(e)$. By averaging argument,
there is a canonical cut $S^*$ such that $\textdiff(S^*,e) >
\frac{\eta}{s(e)} \cdot q(\calB)$. Since $\textdiff(S_e,e)$ is the
only term in the summation that is increased, we have $S^* = S_e$. Therefore, we have %
$$ \frac{\eta}{s(e)}\cdot q(\calB) < \textdiff(S_e,e) = \textref(S_e)
- \textlast(S_e,e) \overset{\ref{item:priority set correctly}}{=} \textref(S_e) - Q_{S_e}.\textpriority(e) +
\frac{\eta'}{s(e)}.$$  Therefore, $\textref(S_e) >
Q_{S_e}.\textpriority(e)$, and so $e \in W_S$ as defined in the
second step. Hence, \textsc{Reset}$(e)$ is invoked. 
\end{proof}

Since \textsc{Reset}$(e)$ is invoked for every violation, we have that \Cref{inv:lazy update}\ref{item:total
  diff not far} is maintained. By design, the second invariant is
trivially maintained whenever \textsc{Reset} is invoked, and also the last
invariant is automatically maintained by the third step. This completes the proof.
\end{proof}

\subsection{Flush}
For each $e \in \bar E$, we call the procedure
$\textsc{Reset}(e)$. Then, we output $\vecw$ which is the same as
$\wreal$.  The total running time is $O(q(\calB) |\bar{E}|)
= \ot( m )$.  
 
\subsection{Total Running Time}

The initialization takes $\ot(m)$. Let $K$ be the number of calls to
\textsc{Punish}$([[C]])$ and let $I$ be the number of calls to
\textsc{Reset}$(e)$ before calling \textsc{Flush}(). The total running time
due to the first step is $O(K\log^2 n) = \ot(K)$, and total running
time due to the second and third steps is $\ot(I)$. It remains to
bound $I$, the total number of calls to \textsc{Reset}$(e)$. Since
each \textsc{Reset}$(e)$ increases of weight $\wreal(e)$ by a factor of
$1+\eta'$, the total number of resets is

\begin{align*}
  O(\sum_{i \in [n]}  \log_{1+O(\eta')} \left(\frac{\wreal(e)}{\vecwinit(e)}\right)) &= O\left(\frac{q(\bset)}{\eta} \cdot  \sum_{i \in [n]} 
                                                 \log ( \frac{\wreal(e)}{\vecwinit(e)})  \right) \\
  &=  \ot \left(\frac{1}{\eta} \cdot  \sum_{i \in [n]} 
  \log (\frac{\wreal(e)}{\vecwinit(e)})  \right).
\end{align*}

%% file: omitted-rounding.tex
\subsection{Proof of~\Cref{lem:spase graph}} \label{sec:spase graph}

It suffices to prove the following lemma. 

\begin{lemma} \label{lem:sparsex}
Given a feasible solution $x$ to kECSS, and a non-negative cost function $  : E \rightarrow \mathbb{R}_{\geq 0}$, and $\epsilon > 0$, there is an algorithm that runs in $\ot(m)$ time, and w.h.p., outputs another feasible solution $y$ to kECSS such that
\begin{itemize}%
\item $\sum_{e \in E} c_ey_e \leq (1+\epsilon) \sum_{e \in E} c_ex_e$.
\item $\support(y) \subseteq \support(x)$.
\item $|\support (y)| = O\left(\frac{k n \log n}{\epsilon^2}\right)$.
\end{itemize}
\end{lemma}

We devote the rest of this subsection to proving~\Cref{lem:sparsex}.

Let $x$ be a near-optimal $k$ECSS fractional solution obtained by Theorem~\ref{thm:intro:fractional $k$ECSS}. 
Compute the solution $y$ using Lemma~\ref{lem:sparsex}. 
Create a graph $G'$ by keeping only edges in the support of $y$.

Before proving the lemma, we first develop an extension to the sparsification theorem from the paper of Benczur and Karger.

We follow the definitions by \cite{BenczurK15, ChekuriQuanrud18}.
\begin{definition} [Edge stength]  Let $G = (V,E,w)$ be a weighted undirected graph.

\begin{itemize}%
\item $G$ is $k$-\textit{connected} if every cut in $G$ has weight at least $k$.
\item A $k$-\textit{strong component} is a maximal non-empty $k$-connected vertex-induced subgraph of $G$.
\item The \textit{strength} of an edge $e$, denoted as $\kappa_e$ is the maximum $k$ such that both endpoints of $e$  belong to some  $k$-strong component.
\end{itemize}
\end{definition}

\begin{lemma}[\cite{BenczurK15}]
$$ \sum_{e \in E} \frac{w_e}{\kappa_e} \leq n - 1 $$
\end{lemma}

\begin{lemma}[\cite{BenczurK15}] \label{lem:approx-strength}
In $\ot(m)$ time, we can compute approximate stength $\tilde \kappa_e$ for each edge $e \in E$ such that $\tilde \kappa_e \leq \kappa_e$ and   $\sum_{e \in E} \frac{w_e}{\tilde \kappa_e} = O(n) $
\end{lemma}

Given a cut $C$ and a subset $S\subseteq C$ of 
its edges, where $|S|\leq k-1$, we say $C\S$ 
is a \emph{constrained cut}. The next theorem states that all constrained cuts would have 
their weights closed to their original weights 
after the sampling.

\begin{theorem} [Extension to Compression Theorem \cite{BenczurK15}] \label{thm:compression-thm}
Given $G = (V,E,w)$, let $p : E \rightarrow [0,1]$ be a probability function over edges of $G$. We construct a random weighted graph  $H = (V,E_H,w')$ as follows.  For each edge $e \in E$, we independently add edge $e$ into $E_H$ with weight $w'_e = w_e/p_e $, with probability $p_e$.  For $\delta \geq \Omega( k d \log n)$, if $p_e \geq \min\{1, \delta \frac{w_e}{\kappa_e} \}$ for all $e \in E$,
then with high probability $\left(\text{over } 1 - \frac{1}{n^d}\right)$,  every constrained cut in $H$ has weight between $(1-\epsilon)$ and $(1+\epsilon)$ times its value in $G$.
\end{theorem}

\begin{proof} 

This theorem follows almost closely the proof of Benczur-Karger.
We sketch here the part where we need a minor modification.

The proof of Benczur-Karger roughly has two components. The first reduces the analysis for general case to the ``weighted sum'' of the ``uniform'' cases where the minimum cut is large, i.e. edge weights are at most $1$ and minimum cut at least $D=  \Omega\left(kd\log n\right)$.
This first component works exactly the same in our case.

Now in each uniform instance which is the second component of Benczur-Karger,  the probabilistic arguments can be made in the following way: For each cut $C$, since edges are sampled independently, we can use Chernoff bound to upper bound the probability that each cut $C$ deviates more than $(1+\epsilon)$ factor (after sampling).
Let $\mu_C$ denote this probability.
Therefore, the bad event that there is a cut deviating too much is upper bounded by $\sum_{C} \mu_C$.

Benczur-Karger analyzes this probability by constructing an auxiliary experiment: Imagine each edge is deleted with probability $p$, then the sum is exactly the expected number of ``empty cuts'' in the resulting graph.
They upper bound this by using the term ${\mathbb E}[2^R]$ where $R$ is the (random) number of connected components in the resulting graph.
They show (using a coupling argument) that ${\mathbb E}[2^R] = O\left(n^2 p^D\right)$, which vanishes whenever $D = \Omega\left(d\log n\right)$.
Here is where we need to slightly change the proof. The bad even that we need to bound is not just all the cuts $\left(\displaystyle \sum_{C} \mu_C\right)$, but also all the constraint cuts. Let $\mu_{C \setminus S}$ be the probability of the bad event that the constraint cut $C \setminus S$ is deviating too much. We want to bound
$$
    \sum_{C} \sum_{S \subseteq C, |S| \leq k-1} \mu_{C \setminus S}.
$$

We will create, by enumerating, $m \choose k$ different graphs $H$ so that each $H$ has at most $k$ edges removed from $G$. Note that all constrained cuts are now defined in these graphs $H$. In the original sampling, if an edge $G$ 
is removed, then we remove it similarly in all graphs $H$ (ignoring it is present in $H$ or not).

Given that there are $R$ connected components in $H$, there are $O(2^R)$ empty cuts.
We consider $m \choose k$ different graphs derived from $H$ by exhaustively remove
a subset $S \subseteq E$ of $k$ edges. Some edges in $S$ might already be removed in $H$, so some configurations will be identical.
We now count the empty cuts in these $m \choose k$ graphs. 
To upper bound $\sum_{C} \sum_{S \subseteq C, |S| \leq k-1} \mu_{C \setminus S}$, we just need to compute the total number of ``empty cuts'' in all these graphs $H$. 

In each $H$, there are at most $R + k$ connected components. Hence, each graph
has at most $2^{R + k}$ empty cuts. Sum up this number among all the graphs,
we get that $$
    \sum_{C} \sum_{S \subseteq C, |S| \leq k-1} \mu_{C \setminus S}
    \leq \mathbb E \left[{m \choose k} 2^{R+k}\right].
$$

Since ${\mathbb E}[2^R] = O(n^2 p^D)$, we get that
$$
\sum_{C} \sum_{S \subseteq C, |S| \leq k} \mu_{C \setminus S} = O\left(
{m \choose k} 2^k n^2 p^D\right) = O\left( {\left ( \frac{2em}{k} \right )}^k n^2 p^D\right),
$$

 which again vanishes if $D$ is large enough (at least $\Omega(k d \log n)$).

\end{proof}

We are now ready to prove \Cref{lem:sparsex}. In fact the same proof in \cite{ChekuriQuanrud18} can be applied once we have \Cref{thm:compression-thm}. 

\begin{proof}[Proof of \Cref{lem:sparsex}]
We first use \Cref{lem:approx-strength} to compute approximate edge strength $\tilde \kappa_e$ for each edge $e \in E$ so that $\tilde \kappa_e \leq \kappa$ and  $\sum_{e \in E} \frac{w_e}{\tilde \kappa_e} = O(n) $ in $\ot (m)$ time.
Let $\delta = \Theta(k d\log n)$ for some large constant $d$.  Let $\text{cost}(x) = \sum_{e \in E} c_e x_e\ $.  For each edge $e \in E$ let $p_e = \min \{1, \frac{ \delta x_e}{\epsilon^2 \tilde \kappa_e} \}$, and $q_e = \min \{1, \frac{\delta c_e x_e}{\epsilon^2 \text{cost}(x) } \} $, and define $r_e = \max (p_e,q_e)$.

We will focus on $x$ from the perspective of kECSS LP with knapsack constraints.

We construct a random graph $H = (V, E', x')$ using $r$ as a probability function over edges of $G$ and we $x$ as weight function of the graph as follows. For each edge $e \in E$, we independently sample edge $e$ into $E'$ with weight $x'_e = x_e/r_e $ with probability $r_e$.    Since
$$
    r_e =  \max (p_e,q_e) \geq p_e =  \min \{1, \frac{ \delta x_e}{\epsilon^2 \tilde \kappa_e} \} \geq  \min\{1, \delta \frac{x_e}{\kappa_e} \}
$$
for sufficiently large constant $d$, by \Cref{thm:compression-thm},
we get w.h.p.,
$$
    \forall C \in \mathcal{C} \forall S \in C, |S| \leq k-1, \sum_{e\in C \setminus S} x'_e \geq (1-\epsilon) \sum_{e\in C \setminus S} x_e \geq (1-\epsilon)(k - |S|).
$$

Observe that $$ \sum_{e \in E} r_e \leq \sum_{e \in E} p_e + \sum_{e \in E} q_e = O( \frac{n\delta}{\epsilon^2}   + \frac{\delta}{\epsilon^2}) = O(\frac{n\delta}{\epsilon^2} )$$

By Chernoff bound, we have  $$ P(\sum_{e\in E} c_e x'_e \geq (1+\epsilon) \sum_{e\in E} c_e x_e ) \leq \exp(-\Omega(\delta)) $$
and,
$$ P( |E'| \geq (1+\epsilon) O( \frac{n\delta}{\epsilon^2})) \leq \exp(-\delta/\epsilon^2) $$
By the union bound, we have the followings  w.h.p.
\begin{align*}
\sum_{e\in C \setminus S} x'_e \geq   (1-\epsilon) (k - |S|), \quad  \forall C \in \mathcal{C}, \forall S \in C, |S| \leq k-1, \\
  |\support(x')| \leq O( \frac{n\delta}{\epsilon^2})\quad \text{ and } 
\sum_{e\in E} c_e x'_e \leq (1+\epsilon)  \sum_{e\in E} c_e x_e
\end{align*}

Therefore, $y' = (1+\epsilon)x'$ is a feasible solution to kECSS. Also,  $ |\support(y')| \leq O( \frac{n\delta}{\epsilon^2})$, and $\sum_{e\in E} c_e y'_e \leq (1+\epsilon)^2 \sum_{e\in E} c_e x_e$. 
Finally, we can get $(1+\epsilon') \sum_{e\in E} c_e x_e$ by a proper  scaling factor for $\epsilon$.

\end{proof}

\subsection{Proof of \Cref{thm: KV}} \label{sec: proof KV}

For the first part of the theorem, 
let $x$ denote the optimal solution in the relaxed LP of $k$ECSS of graph $H$. 
We create a fractional solution $z$ in $D[H]$
as follows: for every edge $e \in E$ in $H$, 
if $e_1$ and $e_2$ are the two opposite directed edges in $D[H]$ derived from $e$, we set 
$z_{e_1}=z_{e_2}=x_e$. It is clear that 
$c(z) = 2c(x)$. We just need to argue that 
$z$ is feasible in the relaxed $k$-arboresences 
problem. Consider a cut $C \in \mathcal{C}$ 
(where $r \in C$ and $C \neq V$). As $x(C) \geq k$, $\sum_{e \in \delta^{+}(C)} z_e \geq k$. 
Furthermore, by Lemma~\ref{lem:KC for box}, $x$ satisfies 
the boxing constraint, that is, 
$0 \leq x_e \leq 1$ for all edges $e \in H$, 
implying that $0 \leq z_e \leq 1$ for all directed edges $e \in D[H]$. This shows that $z$ is feasible and the first part of the theorem is proved. 

For the second part, consider a feasible solution for $k$-arborescences in $D[H]$. If any of the two opposite directed edges is part of the $k$-arborescences, we include its corresponding 
undirected edge in $H$ as part of our induced solution. 
Clearly, the cost of the induced solution cannot be higher and it is a feasible $k$ECSS solution, since it guarantees that the cut value is at least $k$ for all cuts.

\subsection{Polynomially bounded costs}
\label{subsec:bound_cmax}

Since Gabow's algorithm for arborescences has the running time depending on $c_{\max}$, the maximum cost of the edges, we discuss here how to ensure that $c_{\max}$ is polynomially bounded.

Let $x$ be the LP solution obtained from our LP solver. Denote by $C^* = \sum_{e \in E} c_e x_e$, so we have that $C^*$ is between $\opt/2$ and $\opt$, where $\opt$ is the optimal integral value.

First, whenever we see an edge $e \in E$ with $c_e > 2 C^*$, we remove such an edge from the graph $G$.
For each remaining edge $e \in E$, we round the capacity $c_e$ up to the next multiple of $M = \lceil \epsilon C^*/|E| \rceil$.
So, after this rounding up, we have the capacities in $\{M, 2M, \ldots, C^*\}$, and we can then scale them down by a factor of $M$ so that the resulting capacities $c'_e$ are between $1$ and $O(|E|/\epsilon)$.
It is an easy exercise to verify that any $\alpha$-approximation algorithm for $(G,c')$ can be turned into an $\alpha(1+\epsilon)$-approximation algorithm for $(G,c)$.

%% file: mwu-covering.tex
\section{Omitted Proofs}

\subsection{Polynomially Bounded Cost in Proof of \Cref{thm: warmup}} \label{sec:polynomially bounded cost}
Let us assume that the costs $c_e$ are integers (but they can be exponentially large in values).
Karger's sampling~\cite{Karger00mincut} gives a near-linear time algorithm to create a skeleton graph $H$ so that all  cuts in $H$ are approximately preserved, and the minimum cut value is $O(\log |E|)$.
It only requires an easy modification of Karger's arguments to show that we can create a skeleton $H$ such that all $k$-free minimum cuts are approximately preserved, and that the value of the minimum $k$-free cuts is $\Theta(k \log |E|)$.
We will run our static algorithm in graph $H$ instead.
As outlined in Karger's paper~\cite{DBLP:journals/mor/Karger99}, the assumption that we do not know the value of the optimal can be resolved by enumerating them in the geometric scales, and the sampling will guarantee that the running time would not blow up by more than a constant factor.

\subsection{Proof of Theorem~\ref{thm:mwu}}  \label{sec:mwu proof}

The proof is done via duality. The primal and dual solutions will be maintained and updated, until the point where one can argue that their values converge to each other; this implies that both the primal and dual solutions are approximately optimal.  
Recall the primal LP is the covering LP: 
\[\min \{c^T x: A x \geq 1, x \geq 0\} \]
The dual LP is the following packing LP: 
\[\max \{ y^T \ones: y^T A \leq c^T, y \geq 0  \} \]  
For the primal LP, we maintain vectors $\vecw^{(t)} \in {\mathbb R}^n$, where $\vecw^{(0)}_i = 1/c_i$ for each $i \in [n]$. The tentative primal solution on day $t$ is $\bar{\vecw}^{(t)} = \vecw^{(t)}/\minrow(A, \vecw^{(t)})$.
For the dual packing LP, we maintain vectors $\vecf^{(t)} \in {\mathbb R}^m$ where $\vecf^{(0)} = \mathbf{0}$. 
The tentative dual solution on day $t$ is defined as  $\bar{\vecf}^{(t)} = f^{(t)}/\cg(\vecf^{(t)})$, where $\cg(\vecf)$ is the maximum ratio of violated constraints by $\vecf$, that is, \begin{align} \label{def:congestion} \cg(\vecf) = \max_{i \in [n]} \frac{(\vecf^T A)_i}{c_i}.\end{align} Notice that 
$\bar{\vecf}^{(t)}$ is a feasible dual solution on each day.  

Now we explain the update rules on each day. Let $j(t)$ be the row that achieves $A_{j(t)} \vecw^{(t-1)} \leq (1+\epsilon) \minrow(A,\vecw^{(t-1)})$.  
\begin{itemize} 
\item Update $\vecf^{(t)}_{j(t)} \leftarrow \vecf^{(t-1)}_{j(t)} + \delta(t)$ where $\delta(t) = \min_{i \in [n]} \frac{c_i}{A_{j(t), i}}$ is the ``increment'' on day $t$. 

\item Update $\vecw^{(t)}_i \leftarrow \vecw^{(t-1)}_i \exp\left(\epsilon \cdot \frac{\delta(t) A_{j(t),i}}{c_i} \right)$ for each $i \in [n]$.  
\end{itemize}

Denote the primal value at time $t$ by $P(t)= c^T \bar{\vecw}^{(t)}$ and the dual by $D(t) =  ||\bar{\vecf}^{(t)}||_{1}$; so we have $P(t) \geq D(t)$ for all $t$. 
\begin{theorem}
\label{thm:potential}  
Let $t^*$ be the day $t$ for which $P(t)$ is minimized and $N = \Omega(\frac{n}{\epsilon^2} \ln n)$ be the total number of days. 
Then we have that $P(t^*) \leq (1+O(\epsilon)) D(N)$. In particular, $\bar{\vecw}^{(t^*)}$ and $\bar{\vecf}^{(N)}$ are near-optimal primal and dual solutions.   
\end{theorem} 

Our proof relies on the estimates of a potential function defined as $\Phi^{(t)} = c^T \vecw^{(t)}= \sum_{i \in [n]} c_i \vecw^{(t)}_i$. 

\begin{lemma} 
We have, on each day $t$, $$\exp(\epsilon \cdot \cg(\vecf^{(t)}))  \leq \Phi^{(t)} \leq n \cdot \exp  \left( \epsilon (1+3\epsilon) \sum_{0 < t' \leq t} \frac{\delta(t')}{P(t'-1)} \right).$$
\end{lemma}

\begin{proof} 
First we show the lower bound of $\Phi^{(t)}$. Fix column $i \in [n]$ such that $\frac{ ((\vecf^{(t)})^T A)_i  }{c_i} = \cg(\vecf^{(t)})$. Notice that the value of $c_i \vecw^{(t)}_i$ is equal to: 
\[\exp\left( \frac{\epsilon}{c_i} \cdot \sum_{t' \leq t} \delta(t') A_{j(t'), i} \right). \]
 The term $\delta(t') A_{j(t), i}$ is exactly the increase in $((\vecf^{(t)})^T A)_i$ at time $t$, so we have that 
$$c_i \vecw^{(t)}_i \geq  \exp\left( \frac{\epsilon}{c_i} \cdot ((\vecf^{(t)})^T A)_i  \right) = \exp(\epsilon \cdot \cg(\vecf^{(t)})),  $$ as desired. 

Next, we prove the upper bound on the potential function. 
Observe that\footnote{In particular, we use the inequality $e^\gamma \leq 1 + \gamma + \gamma^2$ for $\gamma \in [0,1)$ and the fact that the ratio $\delta(t) A_{j(t),i}/c_i$ is at most $1$.} $\vecw^{(t)}_i \leq \vecw_i^{(t-1)} (1+ \epsilon (1+\epsilon)\cdot \frac{\delta(t) A_{j(t),i } }{c_i})$. 
This formula shows the increase of potential at time $t$ to be at most 
$$\Phi^{(t)} \leq \Phi^{(t-1)} + \sum_{i \in [n]} \epsilon (1+\epsilon) \cdot \delta(t) A_{j(t), i}\vecw^{(t-1)}_i \leq \Phi^{(t-1)} \exp \left( \frac{\epsilon(1+\epsilon) \delta(t)}{\Phi^{(t-1)}} \cdot \sum_{i \in [n]}  A_{j(t), i} \vecw^{(t-1)}_i\right)$$ 
Notice that $\sum_{i \in [n]} A_{j(t),i} \vecw^{(t-1)}_i = (A_{j(t)} \vecw^{(t-1)})$ is at most $(1+\epsilon) \minrow(A, \vecw^{(t-1)})$ by the choice of the update rules. The term reduces further to: 
\[\Phi^{(t)} \leq \Phi^{(t-1)} \exp \left(\frac{\epsilon(1+\epsilon)^2 \delta(t)}{P(t-1)} \right)\leq \Phi^{(t-1)} \exp\left( \frac{\epsilon(1+3 \epsilon) \delta(t)}{P(t-1)}. \right)\]
By applying the fact that $\Phi^{(0)} = n$ and the above fact iteratively, we get the desired bound. %
\end{proof} 

Finally, we argue that the lemma implies Theorem~\ref{thm:potential}. Consider the last day $N$. Taking logarithms on both sides gives us: 
\[\cg(\vecf^{(N)} ) \leq \frac{\ln n}{\epsilon} + (1+3\epsilon) \sum_{0 < t' \leq N} \frac{\delta(t')}{P(t'-1)} \leq \frac{\ln n}{\epsilon} + (1+3\epsilon) \frac{||\vecf^{(N)}||_1}{P(t^*)}\]  
The second inequality uses the fact that $||f^{(N)}||_1 = \sum_{t'} \delta(t')$ and that $P(t^*) \leq P(t)$ for all $t$.

\begin{claim} 
$\cg(\vecf^{(N)}) \geq N/n$, so this implies that $\cg(\vecf^{(N)}) \geq \ln n/ \epsilon^2$ when $N \geq n \ln n/ \epsilon^2$.  
\end{claim} 
\begin{proof}
We will argue that $\sum_{i \in [n]} \frac{\vecf^{(t)}A}{c_i}$ increases by at least one on each day. Since this sum is at most $n \cg(\vecf^{(t)})$, we have the desired result. 
To see the increase, let $i$ be the column that defines $\delta(t)$, that is $i = \arg \min_{i \in [n]} c_i/ A_{j(t), i}$. Notice that $((\vecf^{(t+1)})^T A)_i = ((\vecf^{(t)})^T A)_i+  \delta(t) A_{j(t), i} \geq  ((\vecf^{(t)})^T A)_i+ c_{i}$. This shows an increase of one in the above sum. 
\end{proof} 

Plugging in this term, we have that: 
\[\cg(\vecf^{(N)}) \leq \epsilon \cg(\vecf^{(N)}) + (1+3\epsilon) \frac{||\vecf^{(N)}||_1}{P(t^*)} \]
This implies that $P(t^*) \leq (1+6 \epsilon) D(N)$. 

\subsection{Proof of \Cref{lem:KC for box}}

Let $x$ be a feasible solution $A^{kc} x \geq 1$. Consider $x'_i = \min (x_i,1)$ for each $i \in [n]$. 
We claim that $x'$ satisfies $A x' \geq \kappa$. Consider the constraint $A_j x' \geq \kappa_j$. Let $F =\{i \in \supp(A_j): x_i > 1\}$. If $|F| \geq \kappa_j$, it would imply that $A_j x' \geq \kappa_j$ and we are done. Otherwise, we have $|F| \leq \kappa_j-1$, and the KC constraints guarantee that 
\[\sum_{i \in \supp(A_j)} x'_i = \sum_{i \in \supp(A_j) \setminus F} x_i + |F| \geq \kappa_j  \]

Conversely, let $x$ be a feasible solution $A x \geq \kappa, x \in [0,1]^n$. 
Consider any KC constraint: For any $j \in [m]$ and $F \subseteq \supp(A_j), |F| \leq \kappa_j-1$
\[\sum_{i \in \supp(A_j)\setminus F} x_i = \sum_{i \in \supp(A_j)} x_i - \sum_{i \in F} x_i \geq \kappa_j -|F|\] 
This implies that $x$ itself is feasible for $A^{kc} x \geq 1$.

\subsection{Proof of \Cref{thm:fast LP solver}} \label{sec:fast lp solver full}

By \Cref{lem:KC for box}, it is enough to solve kECSS LP with KC inequalities.

\subsubsection{Interpretation of MWU Framework}  \label{sec:interpretation}

We interpret the analysis in \Cref{sec:mwu proof} in the language of graphs. An interesting feature is that the dual variables are only used in the analysis; it is not used in the implementation at all.

We use $\wreal$ to be the weights that the primal LP maintains. Let $\{(C^{(t)},F^{(t)}, c_{\min}^{(t)})\}_{t \leq T}$ a sequence of normalized free cuts $(C^{(t)},F^{(t)})$ and the value $c_{\min}^{(t)} = \min_{e \in C^{(t)}\setminus F^{(t)}} c(e)$ obtained by the MWU algorithm up to day $T$. For each edge $e$, we define congestion $\cg(e) = \frac{1}{c(e)} \cdot \sum_{t \leq T \colon e \in C^{(t)} \setminus F^{(t)}} c_{\min}^{(t)}$.  The congestion of the graph is denoted as $\cg(G) = \max_{e \in E} \cg(e)$. Note that $\cg(G)$ is precisely the same as $\cg$ in \Cref{def:congestion} when we restrict the LP instance to kECSS LP. Furthermore,  by definition, we have

\begin{align} \label{eq:we and cong}
\forall e,  \wreal(e) \leq \frac{1}{c(e)} \cdot \exp(\eps \cg(G)) 
\end{align}
 
Since the running time of the Range Punisher depends on the change of weights, we need to ensure that the total change (the sum-of-log (SOL) terms) is at most near-linear. We bound the SOL term using a slightly different stopping criteria: Observe that the  analysis rely crucially on the fact that congestion $\cg(G) \geq \frac{1}{\eps^2} \ln m$. We could also use $\cg(G) \geq \frac{1}{\eps^2} \ln m$ as a stopping condition (instead of running up to $O(\frac{1}{\eps^2} m \log m)$ days), and the stopping condition implies the number of days is at most $O(\frac{1}{\eps^2} m \log m)$.

We can infer $\cg(G)$ from the weight function $\wreal$ by the following.  Let $ \phireal(e) :=  \frac{1}{\eps} \cdot \ln (c(e) \cdot \wreal(e))$ for all $e \in E$. By definition of $\cg(e)$, we have $\wreal(e) = \frac{1}{c(e)} \cdot \exp (\eps \cg(e))$, and so $\phireal(e) = \cg(e)$. Therefore, we have
\begin{align}\label{eq:relation to congestion}
\norm{\phireal}_{\infty} = \cg(G).
\end{align}

\subsubsection{Algorithm}

For the implementation, recall that we denote  $\wreal$ to be the real weights on MWU framework, and $\vecw$ to be the approximate weight that the data structure maintains. 

We describe extra bookkeeping from \textsc{RangePunisher} to construct to the final solution. First, it outputs a pair of weight function $(\wreal,\vecw^{\textsol})$ where $\wreal$ is the weights at the end of $\textsc{RangePunisher}$ and $\vecw^{\textsol}= \frac{\vecwinit}{\val_{\vecwinit}(C,F)}$ where $\vecwinit$ is the initial weight function for \textsc{RangePunisher}, and $(C,F)$ is the first normalized mincut obtained during the range punisher.

 Since the range punisher maintains approximate weights, we next explain how to detect the stopping condition using approximate weights. We want to stop as soon as $\norm{\phireal}_{\infty} > \frac{1}{\eps^2}  \cdot \ln m$. Since we can only keep the approximate weights, we can only detect the approximate value with $O(1/\eps)$-additive error as follows.  First, it keeps track of $\phi(e) :=  \frac{1}{\eps} \cdot \ln (c(e) \cdot \wapx(e))$ for all $e \in E$, and early stop as soon as $\norm{\phi}_{\infty} > \frac{1}{\eps} \cdot \ln m$. Since $\vecw$ is $(1+\eps)$-approximation to the real weight $\wreal$, it implies that with respect to weight right before the stopping day, $\norm{\phireal}_{\infty} \leq \frac{1}{\eps} \cdot \ln m + O(\eps^{-1}) = O( \frac{1}{\eps} \ln m)$.

The algorithm for LP solver is described in \Cref{alg:fast LP solver}.

\begin{algorithm}[H]
\KwIn{An undirected graph $G = (V,E)$, a cost function $c$, $\eps \in (0,1)$}
\KwOut{A fractional solution $\wsol$.}
\BlankLine

$\forall e \in E, \wreal(e) \gets \frac{1}{c(e)}$ \;
Let $\tilde \lambda$ be an $(1+\eps)$-approximation to $\opt_{\wreal}$\;
$\lambda \gets \frac{\tilde \lambda }{1+\eps}$\;
$\wbest \gets \frac{\wreal}{\tilde \lambda}$ \; 
\Repeat{$\exists$ \normalfont{a} day such that $\norm{\phi}_{\infty} > \frac{1}{\eps^2} \cdot \ln m$ (and early terminate)}
{
  $(\wreal, \wsol) \gets \textsc{RangePunish}(G,\wreal,\lambda)$\; 
  $\lambda \gets \lambda (1+\eps)$\;
 \lIf{$c^T \wbest > c^T\wsol$}{ $\wbest \gets \wsol$.}
}
%
%
%
%
%
%
%
%
%
%
%
%
%
%
%
%

  \Return{$\vecw^{\textbest}$.}
\caption{\textsc{kECSSLPSolver}($G,c,\eps$)}
\label{alg:fast LP solver}
\end{algorithm}

\subsubsection*{Correctness}
We first show that \Cref{alg:fast LP solver} punish a sequence of $(1+O(\eps))$-approximate normalized free cuts with respect to $\wreal$ where the weight update rule is defined in the \textsc{PunishMin} operations. Initially, $\wreal(e) = \frac{1}{c(e)}$ for all $e \in E$. By definition, $\opt_{\wreal} \in [\tilde \lambda/(1+\eps), \tilde \lambda)$ and thus $\opt_{\wreal} \in [\lambda, (1+\eps)\lambda)$. For each iteration where $\opt_{\wreal} \in [\lambda, (1+\eps)\lambda)$, the range punisher (\Cref{lem:fast range punisher}) keeps punishing $(1+O(\eps))$-approximate normalized free cuts until $\opt_{\wreal} \geq (1+\eps)\lambda$.

By discussion in \Cref{sec:interpretation}, and \Cref{thm:potential}, there must be a day $t^*$ such that in some range such that $\frac{w ^{(t^*)}}{\val_{w^{(t^*)}}(C^{(t^*)},F^{(t^*)})}$ is $(1+O(\eps))$-approximation to the LP solution where $w^{(t^*)}$ is $\wreal$ at day $t^*$.  Since each normalized cut value is within $(1+\eps)$ factor from any other cut inside the same range, we can easily show that the first cut in the range is $(1+\eps)$-competitive with \textit{any} cut in the range. Therefore, \Cref{alg:fast LP solver} outputs $(1+O(\eps))$-approximate solution to kECSS LP. 

\subsubsection*{Running Time} 
 By \Cref{cor:normalized mincut}, the running time for computing the value $\tilde \lambda$ is $\ot(\frac{1}{\eps}\cdot m)$. 
 By \Cref{lem:fast range punisher}, the total running time is $$ \widetilde{O}(m \ell + K+ \frac{1}{\eps} \cdot \sum_{e \in E}\log (\frac{\wreal(e)}{\vecwinit(e)})), $$ where $\ell $ is the number of iterations, and $K$ is the total number of normalized free cuts punished (including all iterations),  $\wreal$ is the final weight at the end of the algorithm, and $\vecwinit(e) = 1/c(e)$ for all $e$.

 Since we early stop as soon as $\norm{\phi}_{\infty} > \frac{1}{\eps^2} \cdot \ln m$,  it means that the day right before we stop we have $\norm{\phireal}_{\infty} = O(\frac{1}{\eps^2} \cdot \ln m)$. By the stopping condition,

 \begin{align} \label{eq:stopping cg}
 \cg(G) \overset{(\ref{eq:relation to congestion})}{=} \norm{\phireal}_{\infty} = O(\frac{1}{\eps^2} \cdot    \ln m). 
 \end{align}

 The following three claims finish the proof.
 
 \begin{claim}
    $\ell = O(\frac{1}{\eps^2} \log m)$.
  \end{claim}
  \begin{proof}
    Initially, we have $\opt_{\wreal} \in [\lambda, (1+\eps)\lambda)$.  By \Cref{eq:we and cong}, we have $\wreal(e) \leq \frac{1}{c(e)} \cdot \exp(\eps \cg(G)) \overset{(\ref{eq:stopping cg})}{=} O(\frac{1}{c(e)} \cdot m^{O(\frac{1}{\eps})})$ for all $e \in E$.  Let $(C^{(0)}, F^{(0)})$ be the first normalized free cut that we punish. Let $\lambda_0$ be the value of that cut. We have that each edge is increase by at most a factor of $m^{O(\frac{1}{\eps})}$, and thus the cut at day right before the stopping happens must be smaller than $\lambda_0 \cdot m^{O(\frac{1}{\eps})}$. Therefore, the number of ranges is $\log_{1+\eps}(m^{O(\frac{1}{\eps})}) = O(\frac{1}{\eps^2} \log m)$.  
  \end{proof}

  \begin{claim}
   $K = O( \frac{1}{\eps^2} m\log m)$.
 \end{claim}
 \begin{proof}
  Observe that for each normalized free cut $(C,F)$ that we punish there exists a bottleneck edge $e \in C \setminus F$ whose $c(e)$ is minimum. By the weight update rule, the congestion is this edge is increased by exactly $1$.  Therefore, the number of normalized free cuts is at most $O(m\cdot \cg(G)) \overset{(\ref{eq:stopping cg})}{=} O(\frac{1}{\eps^2} m\log m)$.  
 \end{proof}

   \begin{claim}
   For each $e$, $\log (\frac{\wreal(e)}{\vecwinit(e)})) = O(\frac{1}{\eps} \log m)$. 
 \end{claim}
 \begin{proof}
   Recall that the initial weight $\vecw^{\text{init}}(e) = 1/c(e)$ for all $e$. Therefore,
   $$ \forall e \in E,   \log (\frac{\wreal(e)}{\vecwinit(e)}))  \overset{(\ref{eq:we and cong})}{\leq}  \eps \cg(G) \overset{(\ref{eq:stopping cg})}{\leq} O(\frac{1}{\eps} \cdot \log m).$$
 \end{proof}